\documentclass[10pt,journal,compsoc]{IEEEtran}

\usepackage{hyperref}
\usepackage{courier}
\usepackage{graphicx}
\usepackage{amsmath}
\usepackage{mathrsfs}
\usepackage{mathtools}
\usepackage{amsthm}
\usepackage{amssymb}
\usepackage{amscd}
\usepackage{times}
\usepackage{caption}
\usepackage{subcaption}
\usepackage{color}
\usepackage{enumerate}

\newtheorem{thm}{Theorem}
\newtheorem{lem}{Lemma}
\newtheorem{prp}{Proposition}
\newtheorem{dfn}{Definition}

\newtheorem{cor}{Corollary}
\newtheorem{ex}{Example}
\newtheorem{remark}{Remark}

\def\E{\mathscr{E}}

\def\T{\mathcal{T}}

\def\Y{\mathcal{Y}}

\def\cS{\mathcal{S}}

\def\H{\mathcal{H}}

\def\Re{\mathbb{R}}
\def\Ce{\mathbb{C}}

\def\U{\mathcal{U}}

\def\V{\mathcal{V}}
\def\Z{\mathcal{Z}}
\def\W{\mathcal{W}}

\def\0{\boldsymbol{0}}
\def\1{\boldsymbol{1}}

\DeclareMathOperator*{\im}{im}

\DeclareMathOperator*{\codim}{codim}

\DeclareMathOperator*{\Span}{Span}

\DeclareMathOperator*{\argmin}{argmin}

\DeclareMathOperator*{\Gr}{Gr}

\DeclareMathOperator*{\rank}{rank}

\DeclarePairedDelimiter\floor{\lfloor}{\rfloor}

\newcommand{\myparagraph}[1]{\smallskip\noindent\textbf{#1.}}

\begin{document} 

\graphicspath{{figures/}}

\title{Homomorphic Sensing}

\author{
    \IEEEauthorblockN{Manolis C. Tsakiris \, \, \,  \, \, \,  Liangzu Peng\\ 
    \vspace{0.2in}   
    \vspace{0.05in}         
    \IEEEcompsocitemizethanks{\IEEEcompsocthanksitem The authors are with the School of Information Science and Technology, ShanghaiTech University, Shanghai, China.
    \IEEEcompsocthanksitem Correspondence to M.C. Tsakiris, mtsakiris@shanghaitech.edu.cn}
} \normalsize}

\IEEEtitleabstractindextext{
\begin{abstract}
A recent line of research termed \emph{unlabeled sensing} and \emph{shuffled linear regression} has been exploring under great generality the recovery of signals from subsampled and permuted measurements; a challenging problem in diverse fields of data science and machine learning. In this paper we introduce an abstraction of this problem which we call \emph{homomorphic sensing}. Given a linear subspace and a finite set of linear transformations we develop an algebraic theory which establishes conditions guaranteeing that points in the subspace are uniquely determined from their homomorphic image under some transformation in the set. As a special case, we recover known conditions for unlabeled sensing, as well as new results and extensions. On the algorithmic level we exhibit two dynamic programming based algorithms, which to the best of our knowledge are the first working solutions for the unlabeled sensing problem for small dimensions. One of them, additionally based on branch-and-bound, when applied to image registration under affine transformations, performs on par with or outperforms state-of-the-art methods on benchmark datasets.
\end{abstract}
}

\maketitle

	\section{Introduction} \label{section:Introduction}
	In a recent line of research termed \emph{unlabeled sensing}, it has been established that uniquely recovering a signal from shuffled and subsampled measurements is possible as long as the number of measurements is at least twice the intrinsic dimension of the signal \cite{Unnikrishnan:TIT18}. The special case where the signal is fully observed but subject to a permutation of its values is known as \emph{shuffled linear regression} \cite{Hsu:NIPS17,Pananjady:TIT18,Choi:ISIT18,Abid:arXiv18}. In its simplest form, it consists of solving a linear system of equations, with the entries of the right hand side vector permuted \cite{Unnikrishnan:Alerton15,Tsakiris:SLR-arXiv18}.
	
	The unlabeled sensing or shuffled linear regression problems and their variations naturally arise in many applications in data science and engineering, such as 1) record linkage for data integration \cite{Lahiri:JASA2005,Shi:arXiv18}, a particularly important problem in medical data analysis where publicly available health records are anonymized, 2) image registration \cite{Lian:PAMI17}, multi-target tracking \cite{POORE20061074} and pose/correspondence estimation \cite{David2004,5459318}, 3) header-free communications in Internet-Of-Things networks \cite{Panajady:ISIT17,Choi:ISIT18,Peng:ICASSP19} and user de-anonymization \cite{4531148,5062142}, 4) acoustic wave field reconstruction \cite{Dokmanic:SPL-19}, 5) system identification under asynchronous input/output samples \cite{1057717}, and many more, e.g., see \cite{Unnikrishnan:TIT18,Pananjady:TIT18}.
	
	\subsection{Prior-Art} \label{subsection:PriorArt}
	\myparagraph{Theory} Suppose that $y = \Pi^* A x^* + \varepsilon \in \Re^m$ is a noisy and shuffled version of some signal $Ax^*$, where $x^* \in \Re^n$ is some unknown regression vector, $\Pi^*$ is some unknown permutation, and $\varepsilon$ is noise. What can be said about the estimation of $x^*$ and $\Pi^*$ given $y$, $A$ and the distribution of $\varepsilon$? This shuffled linear regression problem has been a classic subject of research in the area of record linkage, where predominant methods study maximum likelihood estimators under the working hypothesis that an accurate estimate for the probabilities of transpositions between samples is available \cite{Fellegi:JASA1969,Lahiri:JASA2005}. However, this is a strong hypothesis that does not extend to many applications beyond record linkage.
	
	Very recently important theoretical advances have been made towards understanding this problem in greater generality. Specifically, \cite{Slawski:EJS19,Abid:arXiv17} and \cite{Hsu:NIPS17} have demonstrated that in the absence of any further assumptions, the maximum likelihood estimator $\widehat{x}_{\text{ML}}$ given by
	\begin{align}
	(\widehat{\Pi}_{\text{ML}}, \widehat{x}_{\text{ML}}) = {\argmin_{\Pi, x}} \left\| y - \Pi A x \right\|_2 \label{eq:MLE},
	\end{align} 
	\noindent where $\Pi$ ranges over all permutations, is biased. On the other hand, if the SNR is large enough,  \cite{Pananjady:Allerton16,Pananjady:TIT18} have asserted that $\widehat{\Pi}_{\text{ML}}=\Pi^*$ with high probability. If $\Pi^*$ is sparse enough, i.e., only a small percentage of entries of $Ax^*$ have been shuffled (this is the support of $\Pi^*$), \cite{Slawski:EJS19} have shown that under weaker SNR conditions the supports of $\widehat{\Pi}_{\text{ML}}, \Pi^*$ coincide. Moreover, they provide well behaved error bounds for $\|\widehat{x}_{\text{ML}} - x^*\|_2$ as well as for $\|\widehat{x}_{\text{RR}} - x^*\|_2$, where $\widehat{x}_{\text{RR}}$ is the solution to the convex $\ell_1$ robust regression problem
	\begin{align}
	\min_{x,e} \, \, \|y - Ax -\sqrt{m} e\|_2^2 + m \lambda \|e\|_1, \label{eq:RR} \, \, \, \lambda>0,
	\end{align}  in which the support of the sparse error $e$ is meant to capture the support of the sparse permutation $\Pi^*$.
	
	Another interesting line of work related more to algebraic geometry rather than statistics, is that of \cite{Choi:ISIT18}, which for the noiseless case ($\varepsilon=0$) has proposed the use of symmetric polynomials towards extracting permutation-invariant constraints that $x^* \in \Re^n$ must satisfy. Such a \emph{self-moment} estimator had already been briefly investigated by \cite{Abid:arXiv17} from a statistical point of view, where the authors noted that in the presence of noise it is unclear whether the resulting system of equations has any solutions. Perhaps surprisingly, working with $n$ non-homogeneous polynomials of degrees $1,2,\dots,n$ in $n$ variables, the work of \cite{Tsakiris:SLR-arXiv18} has established that regardless of the value of the noise $\varepsilon$ and under the sole requirement that $A$ is generic\footnote{A rigorous definition of \emph{generic} will be given in \S \ref{subsection:Preliminaries}}, the polynomial system always has a solution and in fact at most $n!$ of them, thus proving the existence of a purely algebraic estimator for $x^*$.
	
	Much less is known for the more challenging and realistic case of unlabeled sensing, where now $y \in \Re^k$ consists of a shuffled noisy subset of the entries of $A x^* \in \Re^m$, i.e., there is no longer a 1-1 correspondence between $y$ and $Ax^*$. The main theoretical finding up to date comes from the seminal work of \cite{Unnikrishnan:TIT18}, according to which, in the absence of noise, $x^*$ is uniquely recoverable from $y$ and $A$ as long as 1) $k \ge 2n$ and 2) $A$ is generic. Inspired by a certain duality between compressed and unlabeled sensing, a recovery condition for noisy data has further been given by \cite{Haghighatshoar:TSP18} in terms of a restricted isometry property. However, this approach is valid only for the special case of $y$ obtained by subsampling $Ax^*$ while maintaining the relative order of the samples.
	
	\myparagraph{Algorithms} Towards computing a solution to the shuffled linear regression problem, which can be solved by brute force in $\mathcal{O}(m!)$, the algorithms presented by \cite{Hsu:NIPS17,Panajady:ISIT17} are conceptually important but applicable only for noiseless data or they have a complexity of at least $\mathcal{O}(m^7)$. When the ratio of shuffled data is small, one may apply the $\ell_1$ robust regression method of \eqref{eq:RR} \cite{Slawski:EJS19}. Other approaches use alternating minimization or multi-start gradient descent to solve \eqref{eq:MLE} \cite{Abid:arXiv17,Abid:arXiv18}, an NP-hard problem for $n>1$ \cite{Pananjady:TIT18}. Due to the high non-convexity such methods are very sensitive to initialization. This is remedied by the algebraically-initialized expectation-maximization method of \cite{Tsakiris:SLR-arXiv18}, which uses the solution to the polynomial system of equations mentioned above to obtain a high-quality initialization. This approach is robust to small levels of noise, efficient for $n \le 5$, and is able to handle fully shuffled data; its main drawback is its exponential complexity in $n$.
	
	In the unlabeled sensing case, which may be thought of as shuffled linear regression with outliers, the above methods in principle break down. Instead, we are aware of only two relevant algorithms, which nevertheless are suitable under strong structural assumptions on the data. The $\mathcal{O}(nm^{n+1})$ method of \cite{Elhami:ICASSP17} applies a brute-force solution, which explicitly relies on the data being noiseless and whose theoretical guarantees require a particular \emph{exponentially spaced} structure on $A$. On the other hand, \cite{Haghighatshoar:TSP18} attempt to solve
	\begin{align}
	{\min_{S, x}} \left\| y - S A x \right\|_2 \label{eq:US},
	\end{align} via alternating minimization, with $S$ in \eqref{eq:US} being a selection matrix\footnote{An \emph{order-preserving selection matrix} is a row-submatrix of the identity matrix. A \emph{selection matrix} is a row-permutation of an order-preserving selection matrix. This is equivalent to a permutation matrix composed by a coordinate projection.}. Their main algorithmic insight is to solve for $S$ given $x$ via dynamic programming. However, their algorithm works only for order-preserving selection matrices $S$, a rather strong limitation, and seems to fail otherwise. It is thus fair to conclude that, to the best of our knowledge, there does not seem to exist a satisfactory algorithm for unlabeled sensing, even for small $n$.

	\subsection{Contributions}
	
	\myparagraph{Theory} In this work we adopt an abstract view of the shuffled linear regression and unlabeled sensing problems, which naturally leads us to a more general formulation that we refer to as \emph{homomorphic sensing}. In homomorphic sensing one is given a finite set $\T$ of linear transformations $\Re^m \rightarrow \Re^m$ (to be called \emph{endomorphisms}) and a linear subspace $\V \subset \Re^m$ of dimension $n$, and asks under what conditions the image $\tau(v)$ of some unknown $v \in \V$ under some unknown $\tau \in \T$ is enough to uniquely determine what $v$ is. This is equivalent to asking under what conditions the relation $\tau_1(v_1)=\tau_2(v_2)$ implies $v_1=v_2$ whenever $\tau_1,\tau_2 \in \T$ and $v_1,v_2 \in \V$. E.g., in shuffled linear regression these endomorphisms are permutations, while in unlabeled sensing they are compositions of permutations with coordinate projections, and the unlabeled sensing theorem of \cite{Unnikrishnan:TIT18} asserts that a sufficient condition for unique recovery is that 1) the coordinate projections preserve at least $2n$ coordinates and 2) $\V$ is generic.
	
	The first theoretical contribution of this paper is a general homomorphic sensing result (Theorem \ref{thm:Endomorphisms}) applicable to arbitrary endomorphisms, and thus of potential interest in a broad spectrum of applications. For generic $\V$, the key condition asks that the dimension $n$ of $\V$ does not exceed the codimension of a certain algebraic variety associated with pairs of endomorphisms from $\T$. This in turn can be used to obtain within a principled framework the unlabeled sensing result of \cite{Unnikrishnan:TIT18}. Our second theoretical contribution is a recovery result (Theorem \ref{thm:Generic}) for generic points in $\V$ (as opposed to all points), which for the unlabeled sensing case says that the coordinate projections need to preserve at least $n+1$ coordinates (as opposed to $2n$).

	\myparagraph{Algorithms} Inspired by \cite{Haghighatshoar:TSP18,Elhami:ICASSP17} and \cite{Yang-PAMI16} we make three algorithmic contributions. First, we introduce a branch-and-bound algorithm for the unlabeled sensing problem by globally minimizing \eqref{eq:US}. Instead of branching over the space of selection matrices, which is known to be intractable \cite{Li:ICCV07}, our algorithm only branches over the space of $x \in \Re^n$, relying on a locally optimal computation of the selection matrix via dynamic programming. Second, it is this dynamic programming feature that also allows us to modify the purely theoretical algorithm of \cite{Elhami:ICASSP17} into a robust and efficient method for small dimensions $n$. These two algorithms constitute to the best of our knowledge the first working solutions for the unlabeled sensing problem. Third, when customized for image registration under an affine transformation, our branch-and-bound algorithm is on par with or outperforms state-of-the-art methods \cite{Lian:PAMI17,Jian:PAMI2011,Myronenko:PAMI2010} on benchmark datasets.

\section{Homomorphic Sensing: Algebraic Theory}
	
The main results of this section are Theorems \ref{thm:Endomorphisms}, \ref{thm:Generic}. To avoid obscuring the main ideas by algebraic arguments exceeding the scope of a computer science paper, we only sketch some proofs and refer the reader to \cite{Tsakiris:ECHS-arXiv18} for the details. All results in this section refer to the noiseless case; an analysis for corrupted data is left to future research. 
	
	\subsection{Preliminaries} \label{subsection:Preliminaries}
	
	For an integer $k$, $[k]=\{1,\dots,k\}$. For non-negative real number $\alpha$, $\floor*{\alpha}$ is the greatest integer $k$ such that $k \le \alpha$.
	
	\subsubsection{$\Re$ versus $\Ce$}
	We work over the complex numbers $\Ce$. This does not contradict the fact that in this paper we are primarily interested in $\Re^m$, rather it facilitates the analysis. E.g., a matrix $T \in \Re^{m \times m}$ may be diagonalizable over $\Ce$ but not over $\Re$. This is the case with permutations, whose eigenvalues are associated with the complex roots of unity. Hence our philosophy is ``check the conditions over $\Ce$, then draw a conclusion over $\Re$"; see Remark \ref{rem:Real2Complex}.
	
	\subsubsection{Abstract Linear Algebra}
	We adopt the terminology of abstract linear algebra \cite{Roman}, since the ideas we discuss in this paper are best delivered in a coordinate-free way. The reader who insists on thinking in terms of matrices may safely replace linear transformations, kernels and images by matrices, nullspaces and rangespaces, respectively.
	
	We work in $\Ce^m$. For a subspace $\V$ we denote by $\dim(\V)$ its dimension. For subspaces $\V,\W$ we say that ``$\V,\W$ do not intersect" if $\V \cap \W = 0$. An endomorphism is a linear transformation $\tau: \Ce^m \rightarrow \Ce^m$; an automorphism is an invertible endomorphism. We denote by $i$ the identity map $i(w)=w, \, \forall w \in \Ce^m$. If $\tau$ is an endomorphism, its kernel $\ker(\tau)$ is the set of all $v \in \Ce^m$ such that $\tau(v)=0$, and its image $\im(\tau)$ is the set of all $\tau(w)$ for $w \in \Ce^m$. By $\rank(\tau)$ we mean $\dim(\im(\tau))$. The preimage of $\tau(v)$ is the set of all $w \in \Ce^m$ such that $\tau(v)=\tau(w)$, i.e., the set of all $v+\xi$, for all $\xi \in \ker(\tau)$. If $\V$ is a linear subspace, by $\tau(\V)$ we mean the set of all vectors $\tau(v)$ for all $v \in \V$. We denote by $\E_{\tau,\lambda}$ the eigenspace of $\tau$ associated to eigenvalue $\lambda$, i.e., the set of all $v \in \Ce^m$ such that $\tau(v) = \lambda v$. For $\tau_1,\tau_2$ endomorphisms the generalized eigenspace of the pair $(\tau_1,\tau_2)$ of eigenvalue $\lambda$ is the set of all $w \in \Ce^m$ for which $\tau_1(w) = \lambda \tau_2(w)$. By a projection $\rho$ of $\Ce^m$ we mean an idempotent ($\rho^2=\rho$) endomorphism. By a coordinate projection we mean a projection that sets to zero certain coordinates of $\Ce^m$ while preserving the rest.
	
	\subsubsection{Algebraic Geometry} \label{subsubsection:AG}
	By an algebraic variety (or variety) of $\Ce^m$ we mean the zero locus of a set of polynomials in $m$ variables. The study of such varieties is facilitated by the use of the Zariski topology, in which every variety is a closed set. In particular, there is a well-developed theory in which topological and algebraic notions of dimension coincide \cite{Hartshorne-1977}. This allows us to assign dimensions to sets such as the intersection of a variety with the complement of another, called quasi-variety. A linear subspace $\cS$ is an algebraic variety and its linear-algebra dimension coincides with its algebraic-geometric dimension.
	The union $\mathcal{A}=\cup_{i \in [\ell]} \cS_i$ of $\ell$ linear subspaces is also an algebraic variety and $\dim(\mathcal{A})=\max_{i \in [\ell]} \dim (\cS_i)$. For a variety $\Y$ which is defined by homogeneous polynomials $\dim(\Y)$ can be characterized as the smallest number of hyperplanes through the origin $0$ that one needs to intersect $\Y$ with to obtain the origin. For $\Y$ a variety of $\Ce^m$, we set $\codim(\Y)=m-\dim(\Y)$. The set of all $n$-dimensional linear subspaces of $\Ce^m$ is itself an algebraic variety of $\Ce^{{m \choose n}}$ called Grassmannian and denoted by
	$\Gr(n,m)$. By a generic subspace $\V$ of dimension $n$ we mean a non-empty open subset $\mathbb{U} \subset \Gr(n,m)$ in the Zariski topology of $\Gr(n,m)$. Such a $\mathbb{U}$ is dense in $\Gr(n,m)$ and if one endows $\Gr(n,m)$ with a continuous probability measure, then $\mathbb{U}$ has measure $1$ \cite{Burgisser:Springer2013}. Hence the reader may safely think of a \emph{generic subspace} as a \emph{random subspace}. When we say ``for a generic $\V$ property $\mathscr{P}$ is true", we mean that the set of all $\V$ for which property $\mathscr{P}$ is true contains a non-empty Zariski open subset of $\Gr(n,m)$.  Hence for a randomly drawn $\V$ property $\mathscr{P}$ will be true with probability $1$. We will make repeated use of the following fact:
	\begin{lem} \label{lem:VarietyIntersectionGeneric}
		Let $\Y$ be a variety defined by homogeneous polynomials and $\V$ a generic linear subspace. Then
		\begin{align}
		\dim (\Y \cap \V) = \max \{\dim(\V) - \codim(\Y),0 \}.
		\end{align}
	\end{lem} A property of the Zariski topology that we need is that the intersection of finitely many non-empty open sets of an irreducible space (such as $\Gr(n,m)$) is open and non-empty.
	
	\subsection{Formulation and first insights} \label{subsection:FirstInsights}
	
	Let $\V \subset \Ce^m$ be a linear subspace of dimension $n$ and let $\T$ be a finite set of endomorphisms of $\Ce^m$. Let $v$ be some point (vector) in $\V$. Suppose that we know the image $\tau(v) \in \Ce^m$ of $v$ for some unspecified $\tau \in \T$. Can we uniquely recover $v$ from knowledge of $\V, \T$ and $\tau(v)$?
	
	\begin{ex}
		In shuffled linear regression $\T$ consists of the $m!$ permutations on the $m$ coordinates of $\Ce^m$. In unlabeled sensing $\T$ consists of the set of all possible combinations of permutations composed with coordinate projections. In both cases $\V$ is the range-space of some matrix $A \in \Ce^{m \times n}$ and $v = A x$ for some $x \in \Ce^n$. The meaning of $A$ and $x$ may vary depending on the application. E.g., in signal processing/control systems $x$ may be the impulse response of some linear filter, while in image registration $x$ may represent the parameters of some affine transformation.
	\end{ex} \noindent The above question motivates the following definition.
	
	\begin{dfn} \label{dfn:UniqueRecovery}
		Let $\V \subset \Ce^m$ be a subspace and $\T$ a finite set of endomorphisms of $\Ce^m$. We say that ``$v_1 \in \V$ is uniquely recoverable in $\V$ under $\T$" if whenever $\tau_1(v_1)=\tau_2(v_2)$ for some $\tau_1,\tau_2 \in \T$ and $v_2 \in \V$, then $v_1=v_2$. If this holds for every $v \in \V$, we have ``unique recovery in $\V$ under $\T$".
	\end{dfn}
	
	\begin{remark} \label{rem:Real2Complex}
		Let $\T$ be a set of endomorphisms of $\Re^m$. These can also be viewed as endomorphisms of $\Ce^m$ (Theorem 2.29 of \cite{Roman}). Let $\V$ be a subspace of $\Re^m$ with basis $v_1,\dots,v_n$ and $\V_{\Ce}=\Span_{\Ce}(v_1,\dots,v_n)$ the subspace of $\Ce^m$ generated by the basis of $\V$. Then unique recovery in $\V_{\Ce}$ under $\T$ implies unique recovery in $\V$ under $\T$.
	\end{remark}
	
	To build some intuition about the notion of unique recovery in $\V$ under $\T$, consider first the case where $\T=\{i,\tau\}$ with  $i$ the identity map and $\tau$ some automorphism. As a first step, we characterize the purely combinatorial condition of Definition \ref{dfn:UniqueRecovery} given in terms of points by the geometric condition of the next proposition given in terms of subspaces.
	
	\begin{lem} \label{lem:SingleTauCharacterization}
		Let $\tau$ be any automorphism of $\Ce^m$ and let $\V$ be any linear subspace. Then we have unique recovery in $\V$ under $\{i,\tau\}$ if and only if $\V \cap \tau(\V) \subset \E_{\tau,1}$.
	\end{lem}
	\begin{proof}
		\emph{(Necessity)} Suppose that $\V \cap \tau(\V) \not\subset \E_{\tau,1}$. Then there exists some $v_1 \in \V \cap \tau(\V)$ not inside $\E_{\tau,1}$. Note that this also implies that $v_1 \not\in \E_{\tau^{-1},1}$. Since $v_1 \in \tau(\V)$ there exists some $v_2 \in \V$ such that $v_1 = \tau(v_2)$. Since $v_1 \not\in \E_{\tau^{-1},1}$ we must have that $v_1 \neq \tau^{-1}(v_1)=v_2$. Hence $v_1 = \tau(v_2)$ with $v_1 \neq v_2$. \emph{(Sufficiency)} Suppose that $v_1 = \tau(v_2)$ for some $v_1,v_2 \in \V$, whence $v_1 \in \V \cap \tau(\V)$. By hypothesis $v_1 \in \E_{\tau,1}$. Hence $\tau(v_1) = v_1$. Hence $\tau(v_1) = \tau(v_2)$. Since $\tau$ is invertible, $v_1=v_2$.
	\end{proof}
	Notice that condition $\V \cap \tau(\V) \subset \E_{\tau,1}$ of Lemma \ref{lem:SingleTauCharacterization} prevents $\V \cap \tau(\V)$ from intersecting $\E_{\tau,\lambda}$ for any $\lambda \neq 1$. Hence, a necessary condition for unique recovery is $\V \cap \tau(\V)$ to not intersect $\E_{\tau,\lambda \neq 1}$. Notice that $\V \cap \tau(\V) \cap \E_{\tau,\lambda} =0$ if and only if $\V \cap \E_{\tau,\lambda}=0$. This places a restriction on the dimension of $\V$, for if $\dim(\V)> \codim(\E_{\tau,\lambda})$ then $\V, \E_{\tau,\lambda}$ will necessarily intersect. Hence, a necessary condition for unique recovery in $\V$ under $\{i,\tau\}$ is
	\begin{align}
	\dim(\V) \le \min_{\lambda \neq 1} \,  \codim  \E_{\tau,\lambda}  \label{eq:SingleTauInvertible-simplified}.
	\end{align} Since the algebraic-geometric dimension of a finite union of linear subspaces is the maximum among the subspace dimensions (\S \ref{subsubsection:AG}), condition \eqref{eq:SingleTauInvertible-simplified} is equivalent to 
	\begin{align}
	\dim(\V) \le \codim\big(\cup_{\lambda \neq 1} \E_{\tau,\lambda} \big) \label{eq:SingleTauInvertible}.
	\end{align} Then for a generic $\V$ satisfying condition \eqref{eq:SingleTauInvertible}, Lemma \ref{lem:VarietyIntersectionGeneric} guarantees that $\V \cap \big(\cup_{\lambda \neq 1} \E_{\tau,\lambda} \big)=0$.
	
	\subsection{Recovery under diagonalizable automorphisms} \label{subsection:DiagonalizableAutomorphisms}
	
	It is not hard to show that when $\tau$ is diagonalizable and $n=\dim(\V)$ is small enough compared to $m$,
	condition \eqref{eq:SingleTauInvertible} is also sufficient for unique recovery in $\V$ under $\{i,\tau\}$\footnote{Lemma \ref{lem:DiagonalizableAutomorphismSingleTau} (with a different proof) is the main insight of the parallel and independent work of  \cite{Dokmanic:SPL-19}. That work studies the same problem as the present paper, but focuses only on diagonalizable automorphisms. On the other hand, it has the advantage that it considers countably many automorphisms, while here we only consider finitely many. }:
	
	\begin{lem} \label{lem:DiagonalizableAutomorphismSingleTau}
		Let $\tau$ be a diagonalizable automorphism of $\Ce^m$ and $\V$ generic subspace, $\dim(\V)\le \floor*{m/2}$. We have unique recovery in $\V$ under $\{i,\tau\}$ if and only if \eqref{eq:SingleTauInvertible} is true.
	\end{lem}
	\begin{proof}
		\emph{} $(\Rightarrow)$ By Lemma \ref{lem:SingleTauCharacterization} $\V \cap \tau(\V) \subset \E_{\tau,1}$. Hence $\V$ does not intersect $\E_{\tau,\lambda}$ for every $\lambda \neq 1$. This implies \eqref{eq:SingleTauInvertible}. $(\Leftarrow)$ Suppose first that $\dim (\E_{\tau,1} ) \le m-n$. The subset $\mathbb{X}$ of $\Gr(n,m)$ on which $\V \cap \tau(\V)=0$ contains the open subset $\mathbb{U}_{\tau}$ on which $\dim(\V + \tau(\V))=2n$. Then $\mathbb{U}_{\tau}$ is non-empty. To see this, note that there exist $2n$ linearly independent eigenvectors $w_1,\dots,w_{2n}$ of $\tau$ such that no more than $n$ of them correspond to the same eigenvalue. Ordering the $w_i$ such that eigenvectors corresponding to the same eigenvalue are placed consecutively, we then define $v_i = w_i+w_{i+n}, \, \forall i \in [n]$. Then $\V=\Span(v_1,\dots,v_n) \in \mathbb{U}_{\tau}$.
		
		Next, suppose that $\dim (\E_{\tau,1} ) > m-n$. Since $n \le \floor*{m/2}$ we have $\dim (\E_{\tau,1} ) \ge n$. Suppose that $v_1 = \tau(v_2)$ for $v_1,v_2 \in \V$. Let $\mathfrak{B}$ be a basis of $\Ce^m$ on which $\tau$ is represented by a diagonal matrix $T \in \Ce^{m \times m}$, and let $V \in \Ce^{m \times n}$ and $V \xi_1, V \xi_2 \in \Ce^m$ be the corresponding representations of a basis of $\V$, and of $v_1,v_2$ respectively, with $\xi_1,\xi_2 \in \Ce^n$. Then the equation $v_1 = \tau(v_2)$ is equivalent to $V \xi_1 = T V \xi_2$. Since $\dim (\E_{\tau,1} ) \ge n$ we may assume without loss of generality that the first $n$ diagonal elements of $T$ are equal to $1$. Letting $V_1 \in \Ce^{n \times n}$ be the top $n \times n$ submatrix of $V$, this implies that $V_1 \xi_1 = V_1 \xi_2$. Then $V_1$ is invertible on a non-empty open subset $\mathbb{U}_{\tau}'$ of $\Gr(n,m)$, on which $v_1=v_2$. Thus $\V \cap \tau(\V) \subset \E_{\tau,1}, \, \forall \V \in \mathbb{U}_{\tau}'$. In conclusion, there is an open set $\mathbb{U} =\mathbb{U}_{\tau}$ or $\mathbb{U} =\mathbb{U}_{\tau}'$, such that for any $\V \in \mathbb{U}$ we have $\V \cap \tau(\V) \subset \E_{\tau,1}$, and so we are done by Lemma \ref{lem:SingleTauCharacterization}. 
	\end{proof}
	
	The extension to multiple automorphisms follows from Lemma \ref{lem:DiagonalizableAutomorphismSingleTau} and the fact that the intersection of finitely many non-empty open sets of $\Gr(n,m)$ is non-empty and open:
	\begin{prp} \label{prp:DiagonalizableAutomorphisms}
		Let $\T$ be a finite set of automorphisms of $\Ce^m$ such that for any $\tau_1,\tau_2 \in \T$ we have that $\tau_1^{-1}\tau_2$ is diagonalizable. Let $\V$ be a generic subspace with $\dim(\V)\le \floor*{m/2}$. We have unique recovery in $\V$ under $\T$ if and only if \eqref{eq:SingleTauInvertible} is true for every $\tau=\tau_1^{-1}\tau_2$ with $\tau_1,\tau_2 \in \T$.
	\end{prp}
	
	Even though the invertibility and diagonalizability requirement of Proposition \ref{prp:DiagonalizableAutomorphisms} may seem too strong, it is satisfied by our canonical example where $\T$ is the set of $m!$ permutations on the $m$ coordinates of $\Ce^m$. In fact, more is true:
	
	\begin{prp} \label{prp:Permutations}
		Let $\T$ be the permutations on the $m$ coordinates of $\Ce^m$. Then $\dim (\E_{\pi,\lambda}) \le m-\floor*{m/2}, \, \forall \pi \in \T, \, \forall \lambda \neq 1$. Hence, for generic subspace $\V $ with $\dim(\V) \le \floor*{m/2}$ we have unique recovery in $\V$ under $\T$.
	\end{prp}
	\begin{proof}
		This follows from basic structural facts about permutations. Let $\pi \in \T$ be a permutation. Then $\pi$ is the product of $c \ge 1$ disjoint cycles, say $\pi = \pi_1\cdots \pi_c$. Suppose that cycle $\pi_i$ cycles $m_i$ coordinates, i.e., it has length $m_i$. Since the cycles are disjoint
		we have $m = \sum_{i=1}^c m_i$. Now, each cycle is diagonalizable with $m_i$ eigenvalues equal to the $m_i$ complex roots of unity, i.e., the roots of the equation $x^{m_i} = 1$. Since the cycles are disjoint, the dimensions of the eigenspaces of $\pi$ are counted additively across cycles. Hence for $\lambda \neq 1$ the dimension of $\E_{\pi,\lambda}$ is at most equal to the number of cycles of length at least $2$. But the number of such cycles is at most $\floor*{m/2}$. Hence $\dim (\E_{\pi,\lambda}) \le \floor*{m/2}$. But $\floor*{m/2} \le m - \floor*{m/2}$, i.e., $\dim (\E_{\pi,\lambda}) \le  m - \floor*{m/2}$. The rest of the statement is a corollary of Proposition \ref{prp:DiagonalizableAutomorphisms}.
	\end{proof}

	\subsection{Recovery under arbitrary endomorphisms} \label{subsection:GeneralResult}
	
	\subsubsection{Unique recovery for all points}
	
	The arguments that led to Proposition \ref{prp:DiagonalizableAutomorphisms} relied heavily on the invertibility of the endomorphisms in $\T$. This is because in that case unique recovery in $\V$ under $\{\tau_1,\tau_2\}$ is equivalent to unique recovery in $\V$ under $\{i,\tau_1^{-1} \tau_2\}$, where $i$ is the identity map. It was this feature that helped us understand the homomorphic sensing property of Definition \ref{dfn:UniqueRecovery} in terms of $\V$ intersecting its image $\tau_1^{-1} \tau_2(\V)$. In turn, the key objects controlling this intersection turned out to be the eigenspaces of $\tau=\tau_1^{-1} \tau_2$ corresponding to eigenvalues different than $1$, as per Lemma \ref{lem:DiagonalizableAutomorphismSingleTau}, whose proof however made explicit use of the diagonalizability of $\tau$. As a consequence, generalizing Proposition \ref{prp:DiagonalizableAutomorphisms} to arbitrary endomorphisms for which $\tau_1$ might not even be invertible, let alone $\tau_1^{-1} \tau_2$ diagonalizable, is not straightforward.
	
	\begin{ex}
		In unlabeled sensing, a permutation composed with a coordinate projection is in general neither invertible nor diagonalizable. E.g., consider a cycle $\pi$ of length $3$, a coordinate projection $\rho$ onto the first two coordinates, and their composition $\rho \pi$:
		\begin{align}
		\pi = \begin{bmatrix}
		0 & 0 & 1 \\
		1 & 0 & 0 \\
		0 & 1 & 0
		\end{bmatrix}, \,
		\rho = \begin{bmatrix}
		1 & 0 & 0 \\
		0 & 1 & 0 \\
		0 & 0 & 0
		\end{bmatrix}, \,
		\rho\pi = \begin{bmatrix}
		0 & 0 & 1 \\
		1 & 0 & 0 \\
		0 & 0 & 0
		\end{bmatrix}. \nonumber
		\end{align} First, $\rank(\rho \pi)=2$ so that $\rho \pi$ is not invertible. Secondly, $\rho \pi$ is nilpotent, i.e., $(\rho \pi)^3=0$, and so the only eigenvalue of $\rho \pi$ is zero. This means that $\rho \pi$ is similar to a $3 \times 3$ Jordan block of eigenvalue $0$, i.e., $\rho \pi$ is far from diagonalizable.
	\end{ex} We proceed by developing two devices. The first one is a generalization of Lemma \ref{lem:DiagonalizableAutomorphismSingleTau} and overcomes the challenge of the potential non-diagonalizability of the endomorphisms.
	
	\begin{lem} \label{lem:AutomorphismSingleTau}
		Let $\V$ be a generic subspace with $\dim(\V) \le \floor*{m/2}$, and $\tau$ any endomorphism of $\Ce^m$ for which \eqref{eq:SingleTauInvertible} is true. Then we have unique recovery in $\V$ under $\{i,\tau\}$.
	\end{lem}
	\begin{proof}
		\emph{(Sketch)} The arguments are similar in spirit with those in the proof of Lemma \ref{lem:DiagonalizableAutomorphismSingleTau} but technically more involved. Let $n=\dim(\V)$. The difficult part is when $\dim(\E_{\tau,1}) \le m-n$, where we prove the existence a non-empty open set $\mathbb{U}$ of $\Gr(n,m)$, such that for every $\V \in \mathbb{U}$ we have $\dim(\V+\tau(\V)) = n+\rank(\tau)$, which is the maximal dimension that the subspace $\V+\tau(\V)$ can have. In analogy with the diagonalizable case, this can be done by working with the Jordan canonical form of $\tau$ and constructing a $\V=\Span(v_1,\dots,v_n) \in \mathbb{U}$ for which the $v_i$ are suitably paired (generalized) eigenvectors of $\tau$.
	\end{proof}
	
	Our second device overcomes the challenge of potential lack of invertibility. We need some notation. Let $\tau_1,\tau_2$ be endomorphisms of $\Ce^m$ and let $\rho$
	be a projection onto $\im(\tau_2)$. Define the variety $\Y_{\rho\tau_1,\tau_2}$ as the set of $w \in \Ce^m$ for which $\rho \tau_1(w)$ and $\tau_2(w)$ are linearly dependent, i.e.,
	\begin{align}
	\Y_{\rho\tau_1,\tau_2} = \big\{w: \, \dim \big(\Span(\rho \tau_1(w),\tau_2(w))\big) \le 1\big\}.
	\end{align} This is indeed a variety because if $\tau_1,\tau_2,\rho$ are represented by matrices $T_1, T_2, P$, then $\Y_{\rho\tau_1,\tau_2}$ is defined by the vanishing of all $2 \times 2$ minors of the matrix $[PT_1w \, \, \, T_2w]$, which are quadratic polynomials in $w$. If $w \in \Y_{\rho\tau_1,\tau_2}$ then there exists some $\lambda_w \in \Ce$ such that
	either $ \tau_2(w) = \lambda_w \rho \tau_1(w)$ or $\rho \tau_1(w) = \lambda_w \tau_2(w)$. Hence $\Y_{\rho\tau_1,\tau_2}$ is the union of all generalized eigenspaces of the endomorphism pairs $(\rho \tau_1,\tau_2)$ and $(\tau_2,\rho \tau_1)$. Note that $\ker(\rho \tau_1-\tau_2)$ is the generalized eigenspace corresponding to eigenvalue $1$, while $\ker(\rho \tau_1), \ker(\tau_2)$ are the generalized eigenspaces of $(\rho \tau_1,\tau_2), (\tau_2,\rho \tau_1)$ respectively, of eigenvalue $0$. In analogy with the automorphism case of Lemma \ref{lem:SingleTauCharacterization} where the eigenspace of eigenvalue $1$ was irrelevant for unique recovery, it turns out that in general the same is true for the generalized eigenspaces of eigenvalues $1$ and $0$. Removing their union
	\begin{align}
	\Z_{\rho\tau_1,\tau_2} = \ker(\tau_2) \cup \ker(\rho \tau_1) \cup \ker(\rho \tau_1-\tau_2),
	\end{align} from $\Y_{\rho\tau_1,\tau_2}$ yields the quasi-variety
	\begin{align}
	\U_{\rho\tau_1,\tau_2}=\Y_{\rho\tau_1,\tau_2} \setminus \Z_{\rho\tau_1,\tau_2}. \label{eq:U}
	\end{align} $\U_{\rho\tau_1,\tau_2}$ plays an analogous role with $\cup_{\lambda \neq 1} \E_{\tau,\lambda}$ when $\tau$ is an automorphism. More precisely, in analogy with \eqref{eq:SingleTauInvertible}, the next theorem shows that the condition that controls homomorphic sensing in general is	
	\begin{align}
	\dim (\V) \le \codim \big(\U_{\rho\tau_1,\tau_2} \big). \label{eq:Udimension}
	\end{align} 
	
	\begin{thm} \label{thm:Endomorphisms}
		Let $\T$ be a finite set of endomorphisms of $\Ce^m$ such that for every $\tau \in \T$
		we have $\rank(\tau) \ge 2n$, for some $n \le \floor*{m/2}$. Then for a general subspace $\V$ of
		dimension $n$, we have unique recovery in $\V$ under $\T$ as long as for every $\tau_1,\tau_2 \in \T$
		there is a projection $\rho$ onto $\im(\tau_2)$ such that for $\U_{\rho\tau_1,\tau_2}$ defined in \eqref{eq:U} condition \eqref{eq:Udimension} is true.
	\end{thm}
	\vspace{-0.3cm}
	\begin{proof} \emph{(Sketch)} The key idea for the case $\T=\{\tau_1,\tau_2\}$ is to view $\V$ as a generic $n$-dimensional subspace of a generic $k$-dimensional subspace $\H$, where $k = \rank(\tau_2)$. Then 
	$\tau_2|_{\H}$ is an isomorphism from $\H$ onto $\im(\tau_2)$, and so unique recovery in $\V$ under $\{\tau_1,\tau_2\}$ follows from unique recovery in $\V$ under $\{i|_{\H},\tau_{\H}\}$, with $\tau_{\H}=\big(\tau_2|_{\H} \big)^{-1} \rho \tau_1|_{\H}$ endomorphism of $\H$. By Lemma \ref{lem:AutomorphismSingleTau} we are done if $\dim \big(\E_{\tau_{\H},\lambda}\big) \le k-n, \, \forall \lambda \neq 1$. Let $\tau_{\H}(w) = \lambda w$, then $\tau_2 \big(\tau_2|_{\H} \big)^{-1} \rho \tau_1(w) = \lambda \tau_2(w)$. Now, $\tau_2 \big(\tau_2|_{\H} \big)^{-1} \rho = \rho$, thus $\rho \tau_1(w) = \lambda \tau_2(w)$. Hence, 
$\E_{\tau_{\H},\lambda}  \subset    \big(\U_{\rho\tau_1,\tau_2} \cap \H\big)$ and the rest follows from dimension considerations. 
\end{proof}

In unlabeled sensing the endomorphisms in $\T$ have the form $\rho \pi$, where $\pi$ is a permutation and $\rho$ is a coordinate projection. Then as per Theorem \ref{thm:Endomorphisms} if $\dim (\V) \le \codim \big(\U_{\rho_2 \rho_1 \pi_1,\rho_2 \pi_2}\big)$ one has unique recovery in $\V$ under $\{\rho_1 \pi_1, \rho_2 \pi_2\}$. Furthermore, via a combinatorial algebraic-geometric argument we obtain a convenient lower bound on $\codim \big(\U_{\rho_2 \rho_1 \pi_1,\rho_2 \pi_2}\big)$:		
	\begin{prp} \label{prp:ProjectionsPermutations}
		Let $\pi_1,\pi_2$ be permutations on the $m$ coordinates of $\Ce^m$ and $\rho_1,\rho_2$ coordinate projections. For $\U_{\rho_2 \rho_1 \pi_1,\rho_2 \pi_2}$ defined in \eqref{eq:U} we have
		\begin{align}
		\floor*{\rank(\rho_2)/2} \le \codim \big(\U_{\rho_2 \rho_1 \pi_1,\rho_2 \pi_2}\big).
		\end{align}
	\end{prp}
	
As a consequence of Theorem \ref{thm:Endomorphisms} and Proposition \ref{prp:ProjectionsPermutations} one has unique recovery in the unlabeled sensing case as long as the dimension of $\V$ does not exceed half the number of the coordinates preserved by each coordinate projection. This is precisely the result of \cite{Unnikrishnan:TIT18} which they obtained by attacking the problem directly via ingenious yet complicated combinatorial arguments. Even though our proof is not necessarily less complicated, it has the advantage of using a framework that generalizes relatively easily. For example, one may consider entry-wise sign corruptions on top of coordinate projections and permutations. In such a case, it is not hard to show that for the same condition as for unlabeled sensing one has unique recovery up to a sign. In general, even though analytically computing $\codim \big(\U_{\rho\tau_1,\tau_2} \big)$ may be challenging, performing this computation in an algebraic geometry software environment such as \texttt{Macaulay2} is in principle straightforward.  
	
	\subsubsection{Unique recovery for generic points}
	
	Often the requirement that \emph{every} $v \in \V$ is uniquely recoverable is unnecessarily strict. Instead, it may be of interest to ask whether unique recovery holds true for a generic $v \in \V$. In such a situation a less demanding technical analysis gives unique recovery under much weaker conditions:
	
	\begin{thm} \label{thm:Generic}
		Let $\T$ be a finite set of endomorphisms of $\Ce^m$. Then a generic point $v$ inside a generic subspace $\V$ of dimension $n$ is uniquely recoverable in $\V$ under $\T$ as long as 1) $\rank(\tau) \ge n+1$ for every $\tau \in \T$, and 2) no two endomorphisms in $\T$ are a scalar multiple of each other.
	\end{thm}
	\vspace{-0.2cm}
	\begin{proof}
		\emph{(Sketch)} Let $V \in \Ce^{m \times n}$ be a basis of $\V$. If $\tau_1(v_1) = \tau_2(v_2)$ then $\tau_2(v_2) \in \tau_1(\V)$ and so $\rank ([T_1V \, \, \, T_2 V \xi]) \le n$ for $\xi \in \Ce^n$ with $v_2=V\xi$. The proof then proceeds by exhibiting, for $\tau_1 \neq \tau_2$, a $\V$ and a $\xi$ for which $\rank ([T_1V \, \, \, T_2 V \xi]) = n+1$. This implies that 
for generic $\V$ and $v_2 \in \V$, $\tau_1(v_1) = \tau_2(v_2)$ for $v_1 \in \V$ only when $\tau_1=\tau_2$. In that case $v_1-v_2 \in \ker(\tau_1)$, and Lemma \ref{lem:VarietyIntersectionGeneric} implies that  $v_1 - v_2=0$.			
	\end{proof}
	\vspace{-0.15cm}
	A consequence of Theorem \ref{thm:Generic} is the unique recovery of a generic vector in the unlabeled sensing case as soon as the coordinate projections preserve at least $n+1$ entries:
	\begin{cor} \label{cor:GenericUnlabeldSensing}
		Let $\T$ be the set of endomorphisms of $\Ce^m$ such that every $\tau \in \T$ has the form $\tau = \rho \pi$, where $\pi$ is a permutation and $\rho$ a coordinate projection. Then for a generic subspace $\V$ of dimension $n$, and a generic $v \in \V$, we have unique recovery of $v$ in $\V$ under $\T$, as long as $\rank(\rho) \ge n+1$ for every $\rho \pi \in \T$.
	\end{cor}

\section{Algorithms and Application}
	
\subsection{Branch-and-bound for unlabeled sensing} \label{subsection:BnB}

	In this section we propose a globally optimal method for the unlabeled sensing problem, by minimizing \eqref{eq:US} via a dynamic-programming based branch-and-bound scheme. Both ingredients are standard, and we just describe how to combine them; see \cite{Emiya:ICASSP2014,Yang-PAMI16} for transparent discussions of branch-and-bound in related contexts. Set $f(x,S)=\left\|y -SAx\right\|_2$, with $x \in \Re^n$ and $S$ a selection matrix. As branching over the space of selection/permutation matrices $S$ is known to be inefficient \cite{Li:ICCV07}, the crucial aspect of our approach is to branch only over the space of $x$, while relying on a local computation of the optimal $S$, say $S_x$, given $x$. Here is where dynamic programming comes into
	play:
	\cite{Haghighatshoar:TSP18} showed that if there exists an
	order-preserving $S_x$ such that
	$f(x,S_x)=\min_{S}\left\|y -SAx\right\|_2$, then $S_x$ can be
	computed via dynamic programming\footnote{That such an assignment problem can be solved via
		dynamic programming was already known by \cite{Aggarwal-FoCS1992}.} at a complexity $\mathcal{O}(mk)$. At
	first sight this does not generalize to any $y,A,x$ as
	none of the minimizers over $S$ is expected to be order-preserving.
	However, if we order $y,Ax$ in descending
	order to obtain say $y^{\downarrow},(Ax)^{\downarrow}$,
	then 1) there is an order-preserving selection matrix $S_x'$ such that
	$\left\|y^{\downarrow}
	-S_x'(Ax)^{\downarrow}\right\|_2=\min_{S}\left\|y^{\downarrow}
	-S(Ax)^{\downarrow}\right\|_2$, and 2) $S_x$ can be easily
	obtained from $S_x'$. In conclusion, given $x$ we can compute $S_x$ in $\mathcal{O}(mk)$; this is in sharp contrast to other linear assignment algorithms such as the Hungarian algorithm, most of which have complexity $\mathcal{O}(m^3)$ \cite{Burkard:2009}. Finally, our strategy becomes that of computing an upper bound of $f$ in a hypercube with center $x_0$ via
	alternating minimization between $x$ and $S$, initialized at $x_0$. Computing a tight lower bound $\ell$ of $f$ for a given hypercube is challenging
	and our choice here is a crude one: $\ell=\left\|y -S_{x_0}Ax_0\right\|_2 - \sigma_1(A)
	\epsilon$, where $\epsilon$ is half the hypercube diagonal and $\sigma_1(A)$ is the largest singular value of $A$. We refer to this as \texttt{Algorithm-A}.
	
\subsection{A robust version of \cite{Elhami:ICASSP17}} \label{subsection:RobustElhami}
It turns out that the dynamic programming \emph{trick} of \S \ref{subsection:BnB} is also the key to a robust version of the theoretical algorithm of \cite{Elhami:ICASSP17}: we randomly select a sub-vector $\bar{y}$ of $y$ of length $n$, and for each $A_i$ out of the $m!/(m-n)!$ many $n \times n$ matrices that can be made by concatenating different rows of $A$ in any order we let $x_i=A_i^{\dagger} \bar{y}$. We then use dynamic programming to select the $x_i$ with the lowest assignment error $ \min_{S}\left\|y -S A x_i\right\|_2$. This is an algorithm of complexity $\mathcal{O}(km^{n+1})$, we call it \texttt{Algorithm-B}.
	
\subsection{Evaluation on synthetic data}
We compare the proposed 1) \texttt{Algorithm-A} of \S \ref{subsection:BnB} and 2) \texttt{Algorithm-B} of \S \ref{subsection:RobustElhami} to other state-of-the-art methods (\S  \ref{subsection:PriorArt}), using normally distributed $A,x,\varepsilon$ with $n=3,m=100$ and $\sigma=0.01$ for the noise. For shuffled linear regression ($k=m$) we compare with 3) \texttt{Slwasky19} that solves \eqref{eq:RR}, 4) \texttt{Tsakiris18} which is the algebraic-geometric method of \cite{Tsakiris:SLR-arXiv18} and 3) \texttt{Abid18} which performs alternating minimization on \eqref{eq:MLE} via least-squares and sorting \cite{Abid:arXiv18}\footnote{The \emph{soft-EM} algorithm of \cite{Abid:arXiv18} consistently fails in our experiment, thus we do not include it in the figure.}. For unlabeled sensing ($k\le m$) we compare with 4) \texttt{Haghighatshoar18} \cite{Haghighatshoar:TSP18}. As seen from Fig. \ref{fig:SLR-US} the proposed methods perform uniformly better and often by a large margin than the other methods, when tested in their robustness against the percentage of shuffled data, outlier ratio and noise level. In particular, we see that for the unlabeled sensing problem of Figs. \ref{figure:US_Shuffles_n3}-\ref{figure:US_sigma_n3} \texttt{Algorithm-A} and \texttt{Algorithm-B} are the only working solutions. Encouraging as these results may be, we do note that an important weakness of these methods is their scalability: for \texttt{Algorithm-B} this is more of an inherent issue due to its brute-force nature, while for \texttt{Algorithm-A} it is due to its naive lower bounding scheme: the consequence of it being far from tight manifests itself at higher dimensions ($n\ge 4$) or large outlier ratios ($k \ll m$), in which the method becomes too slow: it runs\footnote{In this experiment \texttt{Algorithm-A} stops splitting a hypercube when a \emph{depth} $6$ for that hypercube has been reached. Run on an Intel(R) i7-8650U, 1.9GHz, 16GB machine.} in $1$sec for $k=m=100$, $30$sec for $k=80$ and $5$min for $k=50$. In contrast, \texttt{Algorithm-B} is immune to $k$ and runs in about $40$sec. 
		
	\setcounter{figure}{0}
	\begin{figure}[t!]
		\centering
		\begin{minipage}[]{\linewidth}
			\begin{subfigure}[t]{0.49\textwidth}
				\includegraphics[width=\linewidth]{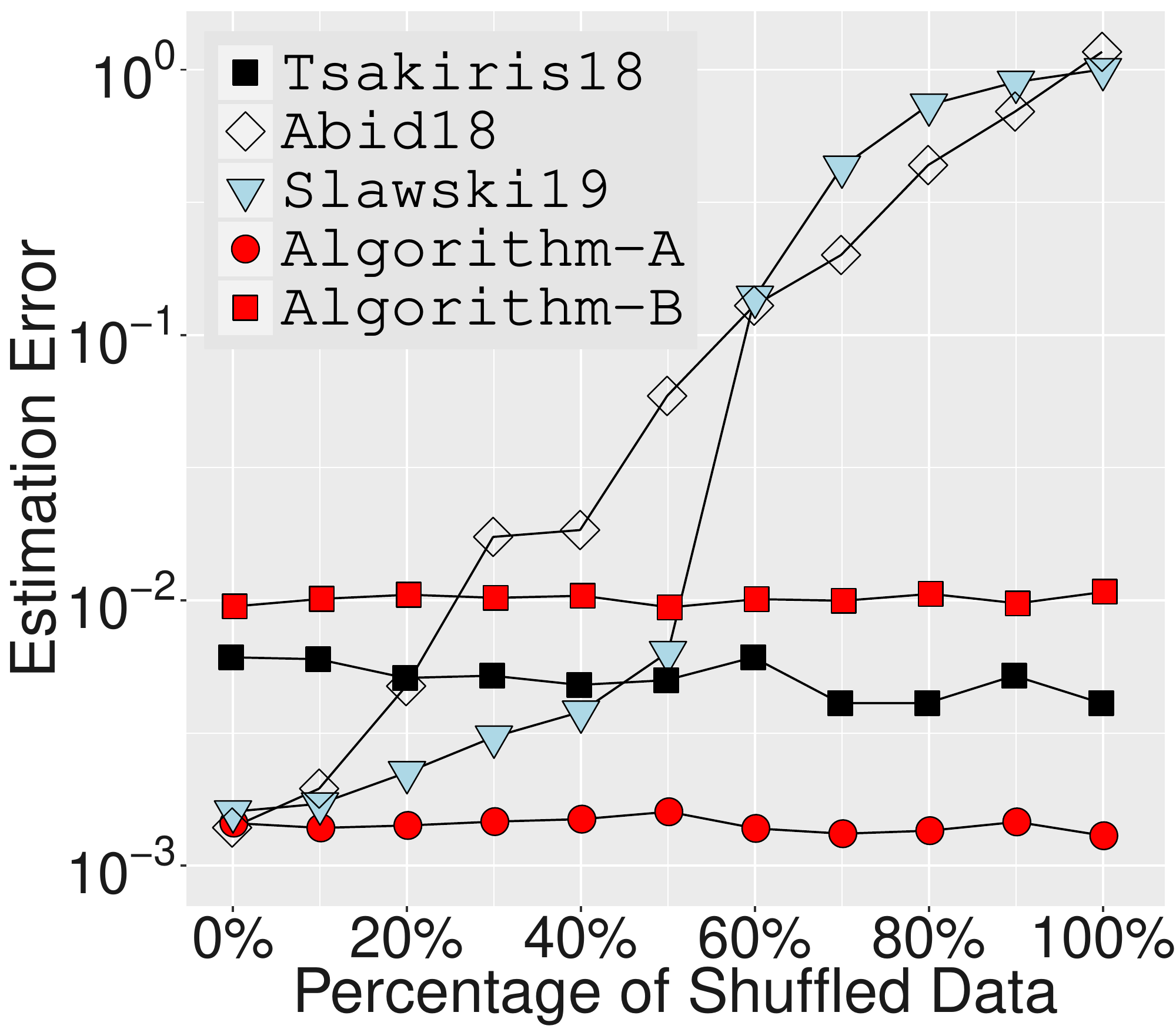}
				\vspace{-0.6cm}
				\caption{shuffled linear regression}\label{figure:SLR_Shuffles_n3}
			\end{subfigure}
			\begin{subfigure}[t]{0.49\textwidth}
				\includegraphics[width=\linewidth]{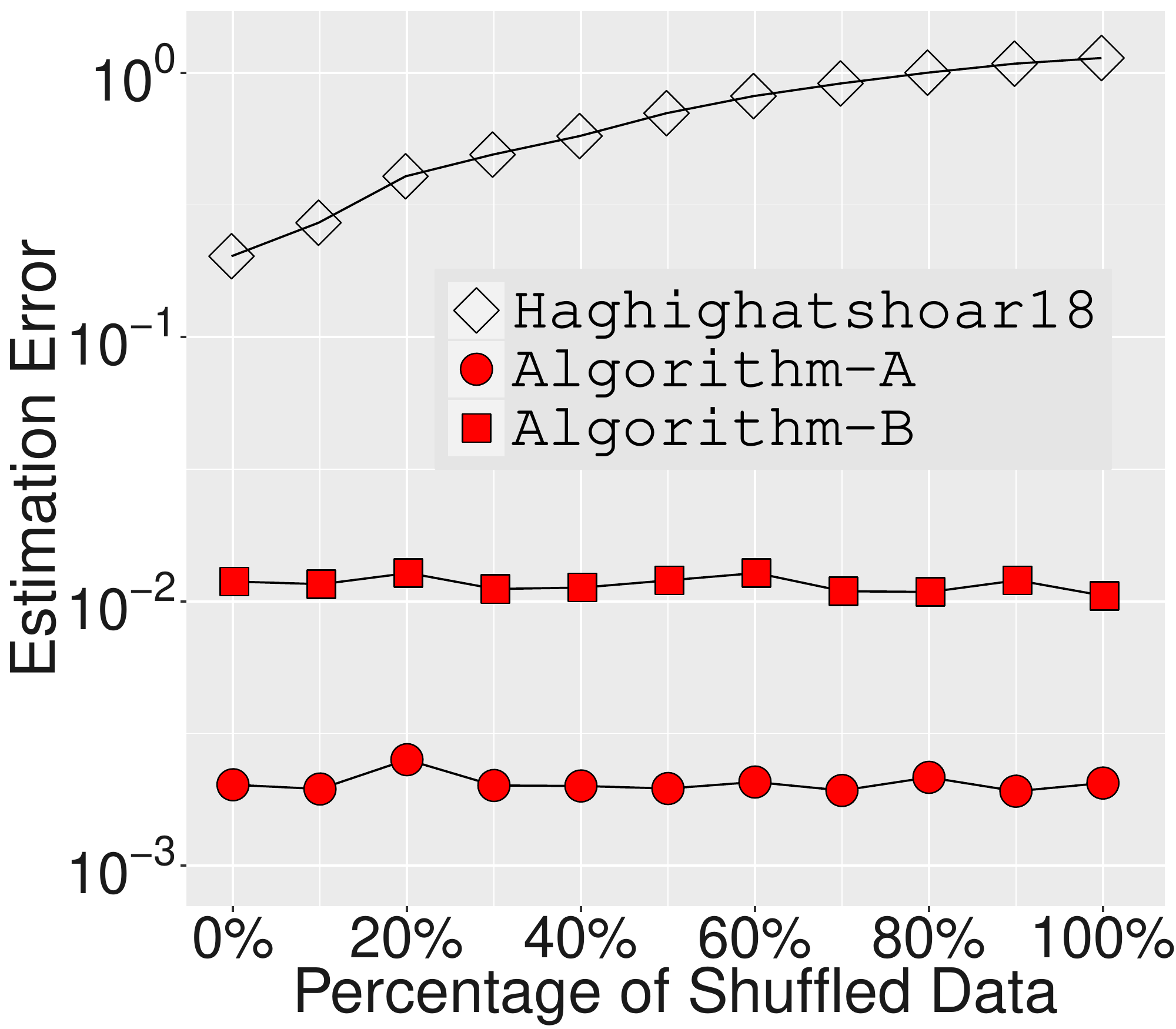}		
				\vspace{-0.6cm}	
				\caption{unlabeled sensing}\label{figure:US_Shuffles_n3}
			\end{subfigure}
		\end{minipage}
		\begin{minipage}[]{\linewidth}
			\begin{subfigure}[t]{0.49\textwidth}
				\includegraphics[width=\linewidth]{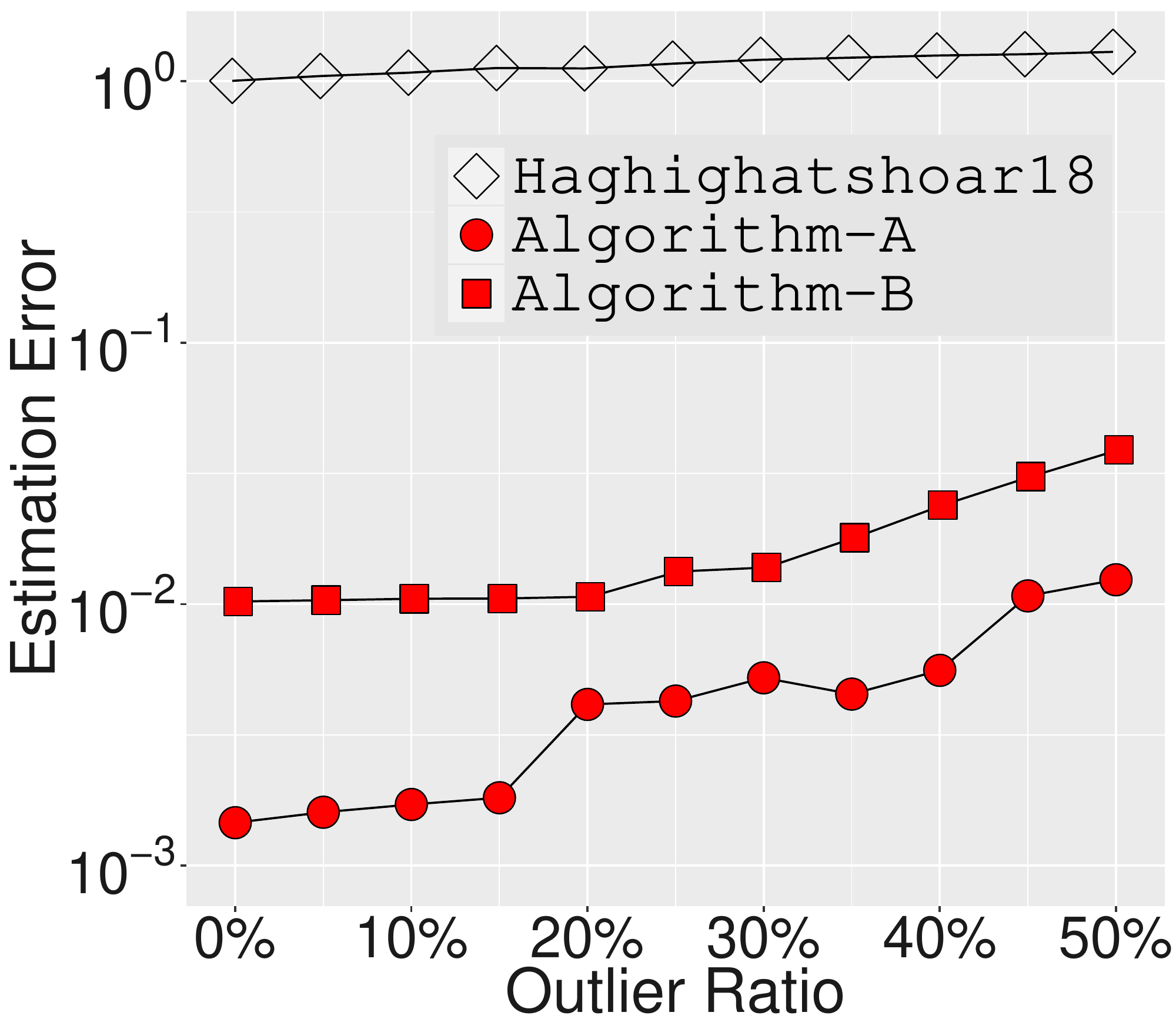}
				\vspace{-0.6cm}	
				\caption{unlabeled sensing}\label{figure:US_k_n3}
			\end{subfigure}
			\begin{subfigure}[t]{0.49\textwidth}
				\includegraphics[width=\linewidth]{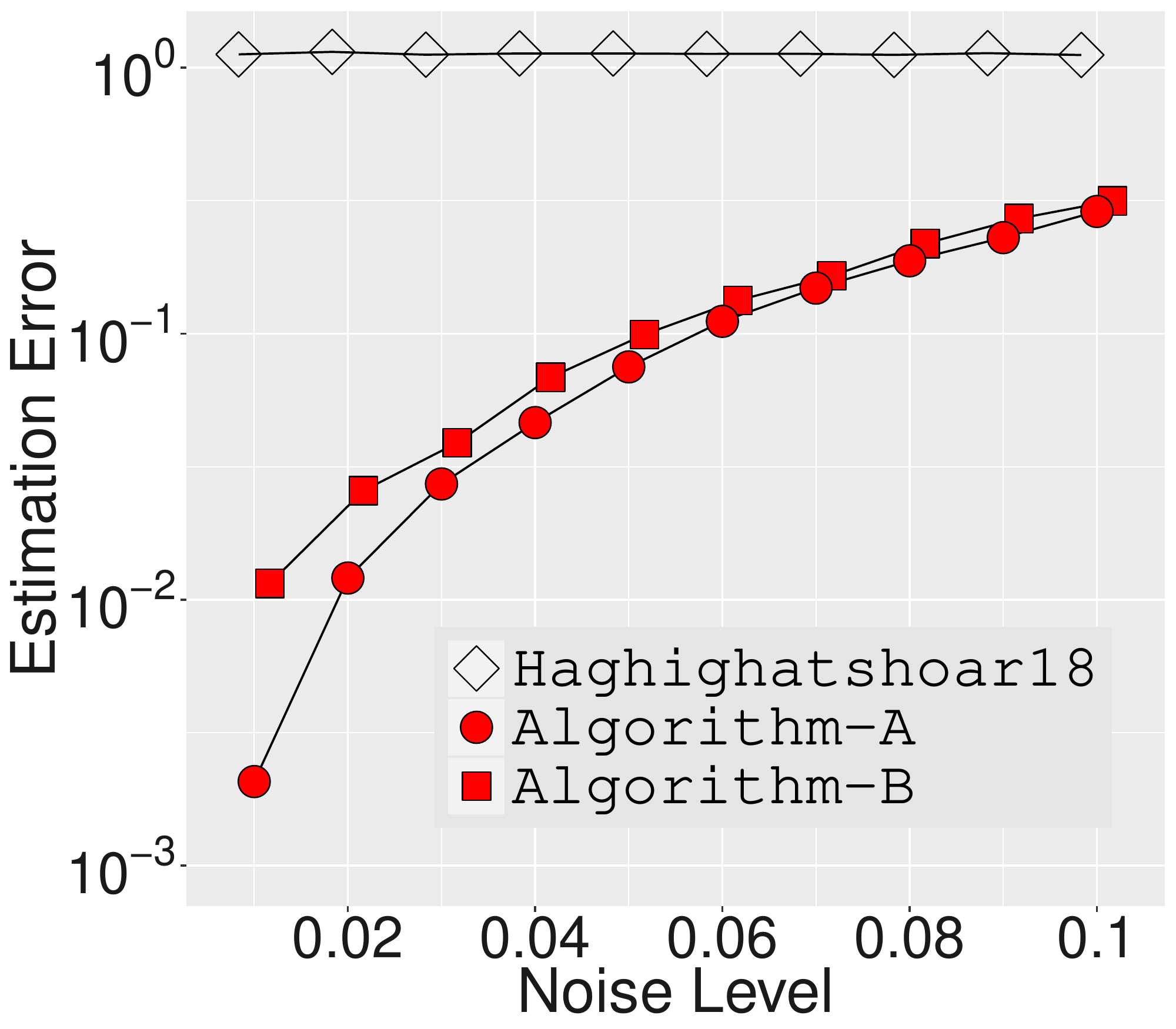}
				\vspace{-0.6cm}	
				\caption{unlabeled sensing}\label{figure:US_sigma_n3}
			\end{subfigure}
		\end{minipage}
		\vspace{-0.3cm}
		\caption{Relative error vs. $\%$ of shuffled data, outlier ratio ($1-k/m$) and noise level ($\sigma$).  $n=3, m=100,k=80,\sigma=0.01$, $1000$ trials. Proposed algorithms in red.}
		\label{fig:SLR-US}
	\end{figure}

	\subsection{Application to image registration}
	\begin{figure}[t]
		\centering
		\begin{minipage}[]{\linewidth}
			\begin{subfigure}[t]{0.49\textwidth}
				\includegraphics[width=\linewidth]{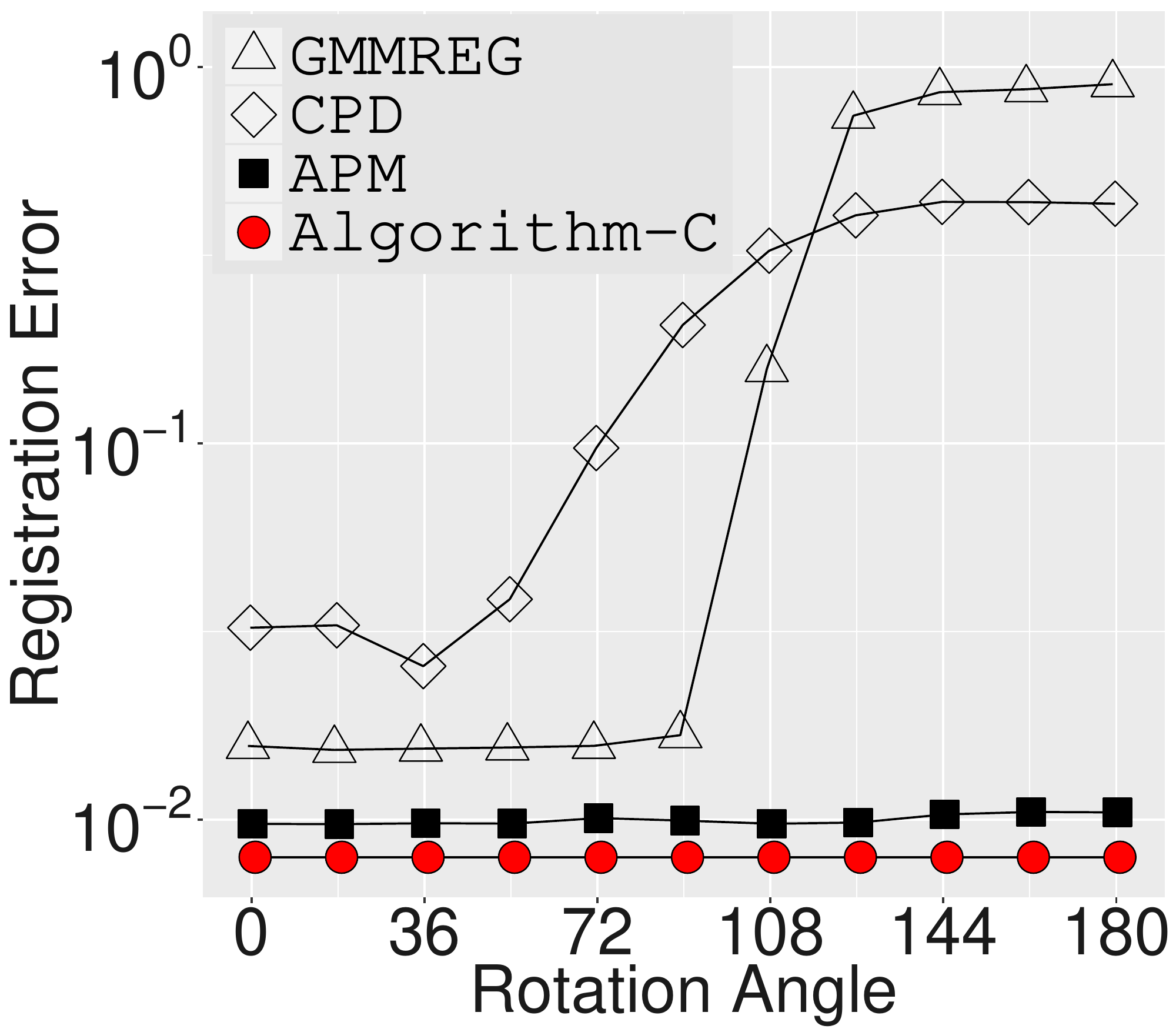}
				\vspace{-0.6cm}
				\caption{2D rotation}\label{figure:Rotation_fu}
			\end{subfigure}
			\begin{subfigure}[t]{0.49\textwidth}
				\includegraphics[width=\linewidth]{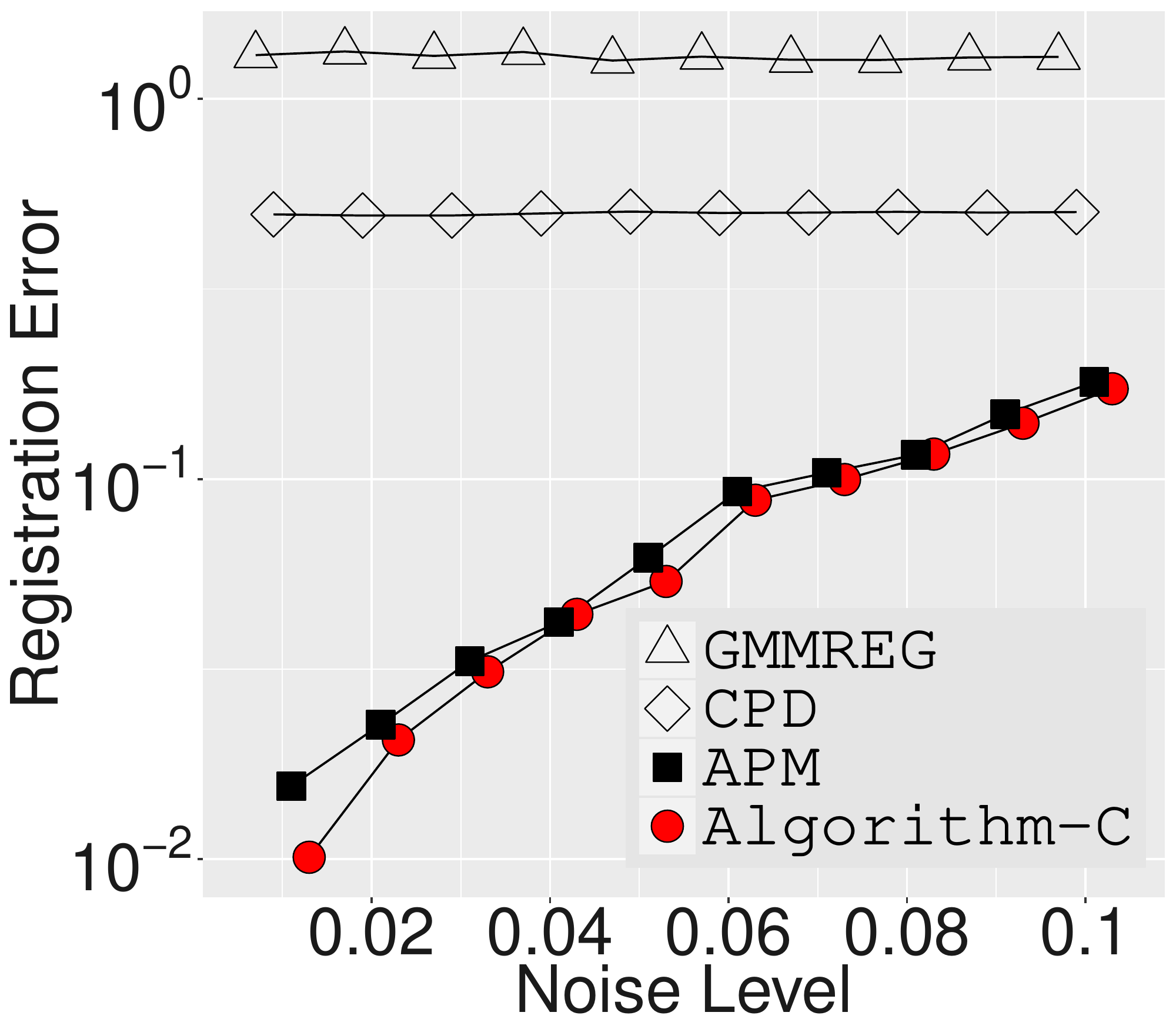}
				\vspace{-0.6cm}
				\caption{affine}\label{figure:sigma_fu}
			\end{subfigure}
		\end{minipage}
		\begin{minipage}{\linewidth}
			\begin{subfigure}[t]{0.49\textwidth}
				\includegraphics[width=\linewidth]{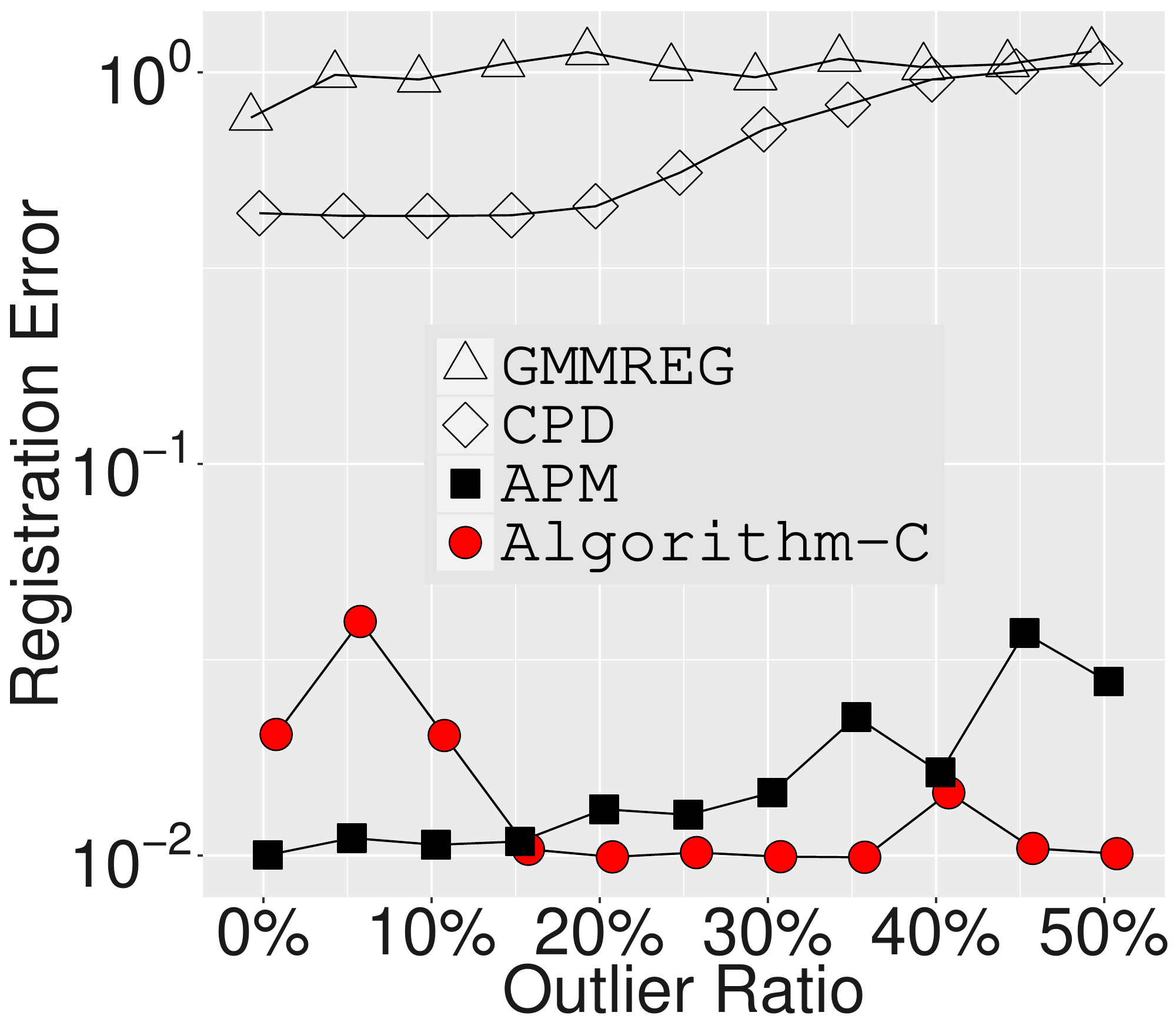}
				\vspace{-0.6cm}
				\caption{affine}\label{figure:k_fu}
			\end{subfigure}
			\begin{subfigure}[t]{0.49\textwidth}
				\includegraphics[width=\linewidth]{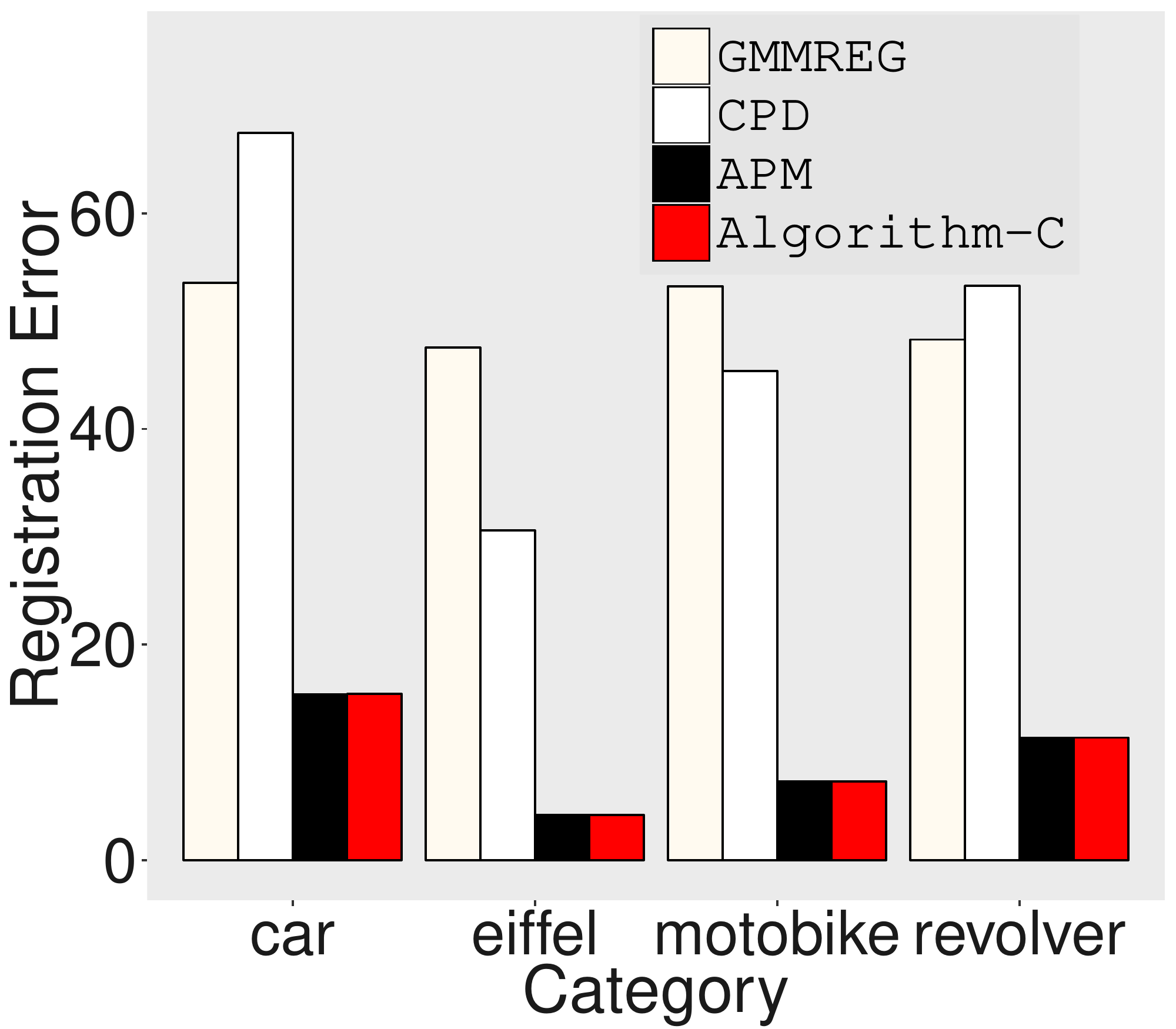}
				\vspace{-0.6cm}
				\caption{real images}\label{figure:CalTechVoc}
			\end{subfigure}
		\end{minipage}
		\begin{minipage}{\linewidth}
			\begin{subfigure}[t]{0.49\textwidth}
				\includegraphics[width=\linewidth]{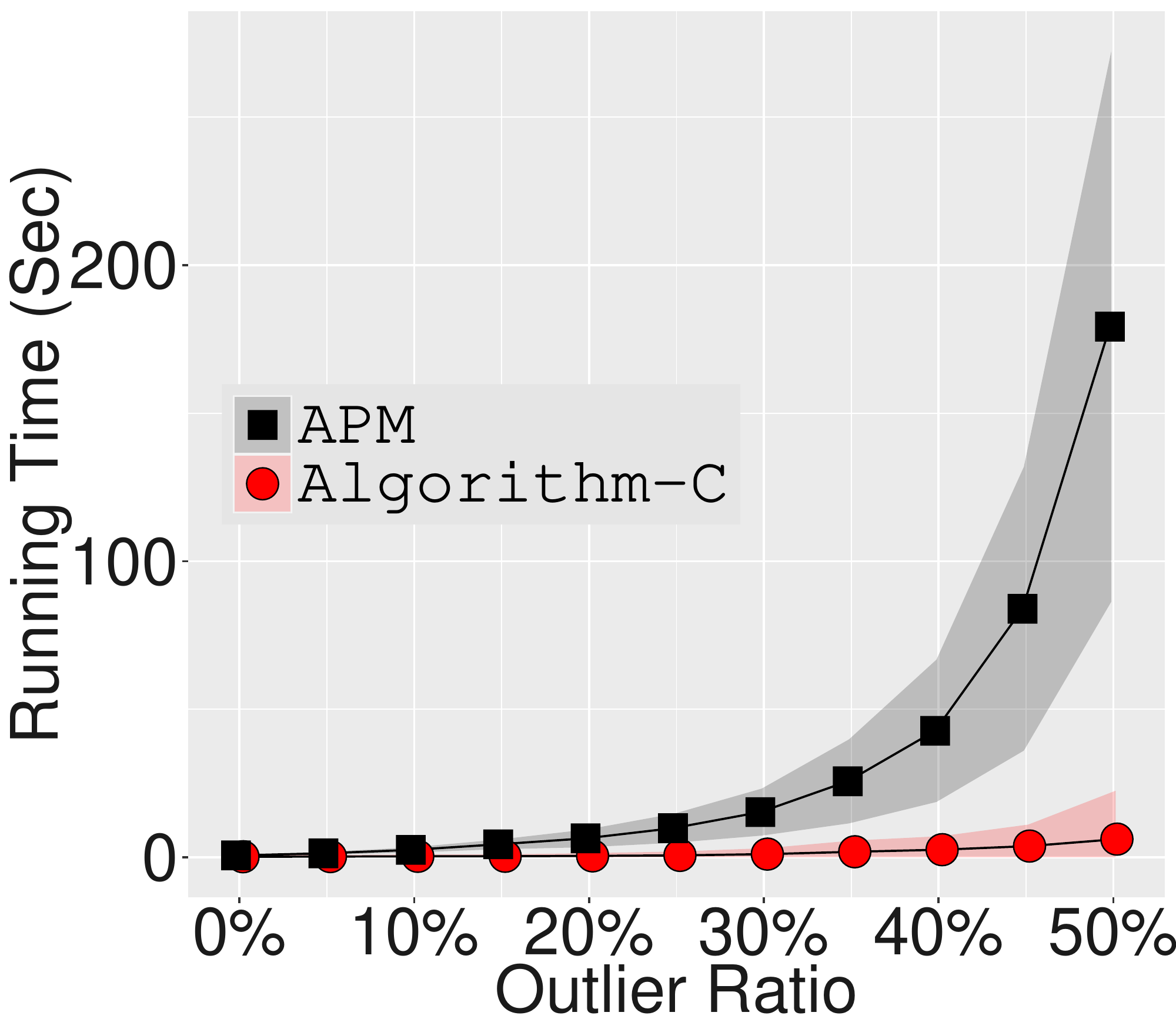}
				\vspace{-0.6cm}
				\caption{affine}\label{figure:RunningTime}
			\end{subfigure}
			\begin{subfigure}[t]{0.49\textwidth}
				\includegraphics[width=\linewidth]{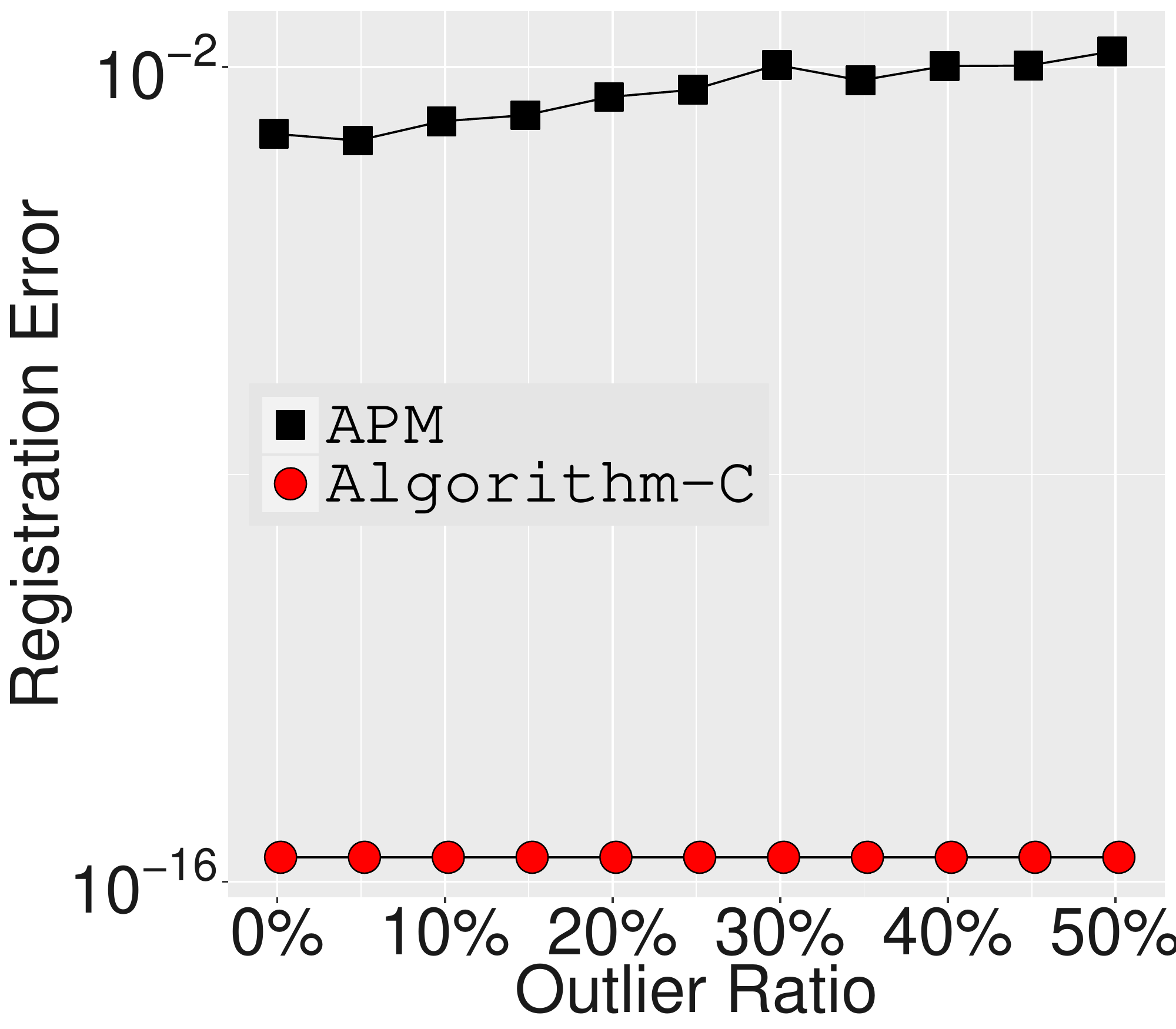}
				\vspace{-0.6cm}
				\caption{affine}\label{figure:RunningTime_Errors}
			\end{subfigure}
		\end{minipage}
		\vspace{-0.3cm}
		\caption{Image registration using synthetic benchmark dataset \emph{The Chinese Character} \cite{Chui-CVIU03} (\ref{figure:Rotation_fu}-\ref{figure:k_fu},\ref{figure:RunningTime}-\ref{figure:RunningTime_Errors},$100$ trials) and the collection of real images used in \cite{Lian:PAMI17} (\ref{figure:CalTechVoc}). Proposed method in red.} \label{fig:Image-Registration}  
	\end{figure}
	Registering point sets $P,Q$ between two images is a classical problem in computer vision. Assuming that $P,Q$ are related by an affine transformation $T$ and that each point in $P$ (\emph{model set}) has a counterpart in $Q$ (\emph{scene set}) \cite{Lian:PAMI17}, jointly searching for the affine transformation and the registration can be done by minimizing the function $F(T,S)=\left\|PT-SQ\right\|_\text{F}$, where $Q \in \Re^{m \times 2}, P \in \Re^{k \times 3}, T \in \Re^{3 \times 2}$, with homogeneous coordinates used for $P$, and $S$ is a selection matrix. This is a matrix version of the unlabeled sensing objective function \eqref{eq:US}. Our contribution here is to adjust algorithm of \S \ref{subsection:BnB} to solve the image registration problem. This involves branching over a $6$-dimensional space to compute $T$, i.e., $n=6$. This does not contradict the remark of the previous section regarding scalability: the key here is that each point correspondence imposes two constraints on $T$ (as opposed to one constraint in the general case), so that, loosely speaking, the \emph{effective} outlier ratio is $1-2k/m$ (as opposed to $1-k/m$). As we will soon see, this has a crucial effect on performance. Finally, as dynamic programming is not applicable to obtain $S$ given $T$, we employ a standard linear assignment algorithm of complexity $\mathcal{O}(m^3)$ \cite{Jonker:C1987}. We refer to this algorithm as \texttt{Algorithm-C}.

We compare 1) \texttt{Algorithm-C} with state-of-the-art image registration techniques, i.e., 2) \texttt{CPD} \cite{Myronenko:PAMI2010}, 3) \texttt{GMMREG} \cite{Jian:PAMI2011}, and 4) \texttt{APM} \cite{Lian:PAMI17}, using a subset of the benchmark datasets used by \cite{Lian:PAMI17}. Since \texttt{APM} is the most competitive among these last three, we let it run to convergence with a tolerance parameter of $0.1$ and set its running time as a time budget for our method; \texttt{CPD} and \texttt{GMMREG} are local methods and they run very fast. When the affine transformation is a rotation (Fig. \ref{figure:Rotation_fu}) \texttt{CPD}, \texttt{GMMREG} only work for small angles, while they fail for general affine transformations (Figs. \ref{figure:sigma_fu}-\ref{figure:k_fu}). On the other hand, \texttt{Algorithm-C} performs comparably to APM with the following twist in Fig. \ref{figure:k_fu}: when the outlier ratio is small, \texttt{APM} converges very quickly resulting to an inadequate time budget for our method. Conversely, when the outlier ratio is large, \texttt{APM}'s accuracy becomes inadequate while it is slow enough for our method to perform even better than for fewer outliers. These running times for \texttt{APM} are shown in Fig. \ref{figure:RunningTime} where the same experiment is run for noiseless data: \texttt{APM} still uses a tolerance of $0.1$ while to make a point we set the tolerance of our method to zero and let it terminate. As seen in Figs. \ref{figure:RunningTime}-\ref{figure:RunningTime_Errors}, our method terminates significantly faster than \texttt{APM} for large outlier ratios, suggesting that its branch-and-bound structure may have an advantage over that of \texttt{APM}. 

\section*{Acknowledgement}
\noindent The first author is grateful to Dr. Aldo Conca for first noting and sharing the phenomenon of Lemmas \ref{lem:SingleTauCharacterization} and \ref{lem:DiagonalizableAutomorphismSingleTau}. We also thank Dr. Laurent Kneip for 
suggesting the case study of image registration under general affine transformations.

	\bibliographystyle{IEEEtran}
	\bibliography{HS-arXiv28Apr19}

\begin{thebibliography}{10}
\providecommand{\url}[1]{#1}
\csname url@samestyle\endcsname
\providecommand{\newblock}{\relax}
\providecommand{\bibinfo}[2]{#2}
\providecommand{\BIBentrySTDinterwordspacing}{\spaceskip=0pt\relax}
\providecommand{\BIBentryALTinterwordstretchfactor}{4}
\providecommand{\BIBentryALTinterwordspacing}{\spaceskip=\fontdimen2\font plus
\BIBentryALTinterwordstretchfactor\fontdimen3\font minus
  \fontdimen4\font\relax}
\providecommand{\BIBforeignlanguage}[2]{{%
\expandafter\ifx\csname l@#1\endcsname\relax
\typeout{** WARNING: IEEEtran.bst: No hyphenation pattern has been}%
\typeout{** loaded for the language `#1'. Using the pattern for}%
\typeout{** the default language instead.}%
\else
\language=\csname l@#1\endcsname
\fi
#2}}
\providecommand{\BIBdecl}{\relax}
\BIBdecl

\bibitem{Unnikrishnan:TIT18}
J.~Unnikrishnan, S.~Haghighatshoar, and M.~Vetterli, ``Unlabeled sensing with
  random linear measurements,'' \emph{IEEE Transactions on Information Theory},
  vol.~64, no.~5, pp. 3237--3253, 2018.

\bibitem{Hsu:NIPS17}
D.~J. Hsu, K.~Shi, and X.~Sun, ``Linear regression without correspondence,'' in
  \emph{Advances in Neural Information Processing Systems (NIPS)}, 2017, pp.
  1531--1540.

\bibitem{Pananjady:TIT18}
A.~Pananjady, M.~J. Wainwright, and T.~A. Courtade, ``Linear regression with
  shuffled data: Statistical and computational limits of permutation
  recovery,'' \emph{IEEE Transactions on Information Theory}, vol.~64, no.~5,
  pp. 3286--3300, 2018.

\bibitem{Choi:ISIT18}
X.~Song, H.~Choi, and Y.~Shi, ``Permuted linear model for header-free
  communication via symmetric polynomials,'' in \emph{International Symposium
  on Information Theory (ISIT)}, 2018.

\bibitem{Abid:arXiv18}
A.~Abid and J.~Zou, ``Stochastic em for shuffled linear regression,''
  arXiv:1804.00681v1 [stat.ML], Tech. Rep., 2018.

\bibitem{Unnikrishnan:Alerton15}
J.~Unnikrishnan, S.~Haghighatshoar, and M.~Vetterli, ``Unlabeled sensing:
  solving a linear system with unordered measurements,'' in \emph{53rd IEEE
  Annual Allerton Conference on Communication, Control and Computing}, 2015.

\bibitem{Tsakiris:SLR-arXiv18}
M.~C. Tsakiris, L.~Peng, A.~Conca, L.~Kneip, Y.~Shi, and H.~Choi, ``An
  algebraic-geometric approach to shuffled linear regression,''
  arXiv:1810.05440v1 [cs.LG], Tech. Rep., 2018.

\bibitem{Lahiri:JASA2005}
P.~Lahiri and M.~D. Larsen, ``Regression analysis with linked data,''
  \emph{Journal of the American Statistical Association}, vol. 100, no. 469,
  pp. 222--230, 2005.

\bibitem{Shi:arXiv18}
X.~Shi, X.~Li, and T.~Cai, ``Spherical regression under mismatch corruption
  with application to automated knowledge translation,'' arXiv:1810.05679v1
  [stat.ME], Tech. Rep., 2018.

\bibitem{Lian:PAMI17}
W.~Lian, L.~Zhang, and M.-H. Yang, ``An efficient globally optimal algorithm
  for asymmetric point mathcing,'' \emph{IEEE Transactions on Pattern Analysis
  and Machine Intelligence}, vol.~39, no.~7, pp. 1281--1293, 2017.

\bibitem{POORE20061074}
A.~B. Poore and S.~Gadaleta, ``Some assignment problems arising from multiple
  target tracking,'' \emph{Mathematical and Computer Modelling}, vol.~43,
  no.~9, pp. 1074 -- 1091, 2006.

\bibitem{David2004}
P.~David, D.~DeMenthon, R.~Duraiswami, and H.~Samet, ``Softposit: Simultaneous
  pose and correspondence determination,'' \emph{International Journal of
  Computer Vision}, vol.~59, no.~3, pp. 259--284, 2004.

\bibitem{5459318}
M.~Marques, M.~Stosic, and J.~Costeira, ``Subspace matching: Unique solution to
  point matching with geometric constraints,'' in \emph{International
  Conference on Computer Vision (ICCV)}, 2009, pp. 1288--1294.

\bibitem{Panajady:ISIT17}
A.~Pananjady, M.~J. Wainwright, and T.~A. Courtade, ``Denoising linear models
  with permuted data,'' in \emph{IEEE International Symposium on Information
  Theory (ISIT)}, 2017, pp. 446--450.

\bibitem{Peng:ICASSP19}
L.~Peng, X.~Song, M.~Tsakiris, H.~Choi, L.~Kneip, and Y.~Shi,
  ``Algebraically-initialized expectation maximization for header-free
  communication,'' in \emph{IEEE International Conference on Acoustics, Speech
  and Signal Processing (ICASSP)}, 2019, pp. 5182--5186.

\bibitem{4531148}
A.~Narayanan and V.~Shmatikov, ``Robust de-anonymization of large sparse
  datasets,'' in \emph{IEEE Symposium on Security and Privacy}, 2008, pp.
  111--125.

\bibitem{5062142}
L.~Keller, M.~J. Siavoshani, C.~Fragouli, K.~Argyraki, and S.~Diggavi,
  ``Identity aware sensor networks,'' in \emph{IEEE INFOCOM}, 2009, pp.
  2177--2185.

\bibitem{Dokmanic:SPL-19}
I.~Dokmani{\'c}, ``Permutations unlabeled beyond sampling unknown,'' \emph{IEEE
  Signal Processing Letters}, vol.~26, no.~6, pp. 823--827, 2019.

\bibitem{1057717}
A.~Balakrishnan, ``On the problem of time jitter in sampling,'' \emph{IRE
  Transactions on Information Theory}, vol.~8, no.~3, pp. 226--236, 1962.

\bibitem{Fellegi:JASA1969}
I.~P. Fellegi and A.~B. Sunter, ``A theory for record linkage,'' \emph{Journal
  of the American Statistical Association}, vol.~64, no. 328, pp. 1183--1210,
  1969.

\bibitem{Slawski:EJS19}
M.~Slawski and E.~Ben-David, ``Linear regression with sparsely permuted data,''
  \emph{Electronic Journal of Statistics}, vol.~13, no.~1, pp. 1--36, 2019.

\bibitem{Abid:arXiv17}
A.~Abid, A.~Poon, and J.~Zou, ``Linear regression with shuffled labels,''
  arXiv:1705.01342v2 [stat.ML], Tech. Rep., 2017.

\bibitem{Pananjady:Allerton16}
A.~Pananjady, M.~J. Wainwright, and T.~A. Courtade, ``Linear regression with
  shuffled data: Statistical and computational limits of permutation
  recovery,'' in \emph{54th IEEE Annual Allterton Conference on Communication,
  Control and Computing}, 2016, pp. 417--424.

\bibitem{Haghighatshoar:TSP18}
S.~Haghighatshoar and G.~Caire, ``Signal recovery from unlabeled samples,''
  \emph{IEEE Transactions on Signal Processing}, vol.~66, no.~5, pp.
  1242--1257, 2018.

\bibitem{Elhami:ICASSP17}
G.~Elhami, A.~Scholefield, B.~B. Haro, and M.~Vetterli, ``Unlabeled sensing:
  Reconstruction algorithm and theoretical guarantees,'' in \emph{IEEE
  International Conference on Acoustics, Speech and Signal Processing
  (ICASSP)}, 2017, pp. 4566--4570.

\bibitem{Yang-PAMI16}
J.~Yang, H.~Li, D.~Campbell, and Y.~Jia, ``Go-icp: A globally optimal solution
  to 3d icp point-set registration,'' \emph{IEEE Transactions on Pattern
  Analysis and Machine Intelligence}, vol.~38, no.~11, pp. 2241--2254, 2016.

\bibitem{Li:ICCV07}
H.~Li and R.~Hartley, ``The 3d-3d registration problem revisited,'' in
  \emph{International Conference on Computer Vision (ICCV)}, 2007, pp. 1--8.

\bibitem{Jian:PAMI2011}
B.~Jian and B.~C. Vemuri, ``Robust point set registration using gaussian
  mixture models,'' \emph{IEEE Transactions on Pattern Analysis and Machine
  Intelligence}, vol.~33, no.~8, pp. 1633--1645, 2011.

\bibitem{Myronenko:PAMI2010}
A.~Myronenko and X.~Song, ``Point set registration: coherent point drift,''
  \emph{IEEE Transactions on Pattern Analysis and Machine Intelligence},
  vol.~32, no.~12, pp. 2262--2275, 2010.

\bibitem{Tsakiris:ECHS-arXiv18}
M.~C. Tsakiris, ``Eigenspace conditions for homomorphic sensing,''
  arXiv:1812.07966v3 [math.CO], Tech. Rep., 2019.

\bibitem{Roman}
S.~Roman, \emph{Advanced Linear Algebra}.\hskip 1em plus 0.5em minus
  0.4em\relax Springer, 2008.

\bibitem{Hartshorne-1977}
R.~Hartshorne, \emph{Algebraic Geometry}.\hskip 1em plus 0.5em minus
  0.4em\relax Springer, 1977.

\bibitem{Burgisser:Springer2013}
P.~B{\"u}rgisser and F.~Cucker, ``The geometry of numerical algorithms,''
  \emph{Springer Science \& Business Media}, vol. 349, 2013.

\bibitem{Emiya:ICASSP2014}
V.~Emiya, A.~Bonnefoy, L.~Daudet, and R.~Gribonval, ``Compressed sensing with
  unknown sensor permutation,'' in \emph{IEEE International Conference on
  Acoustics, Speech and Signal Processing (ICASSP)}, 2014, pp. 1040--1044.

\bibitem{Aggarwal-FoCS1992}
A.~Aggarwal, A.~Bar-Noy, S.~Khuller, D.~Kravets, and B.~Schieber, ``Efficient
  minimum cost matching using quadrangle inequality,'' in \emph{Annual
  Symposium on Foundations of Computer Science}, 1992, pp. 583--592.

\bibitem{Burkard:2009}
R.~E. Burkard, M.~Dell'Amico, and S.~Martello, \emph{Assignment
  Problems}.\hskip 1em plus 0.5em minus 0.4em\relax Springer, 2009.

\bibitem{Chui-CVIU03}
H.~Chui and A.~Rangarajan, ``A new point matching algorithm for non-rigid
  registration,'' \emph{Comput. Vis. Image Underst.}, vol.~89, no. 2-3, pp.
  114--141, 2003.

\bibitem{Jonker:C1987}
R.~Jonker and A.~Volgenant, ``A shortest augmenting path algorithm for dense
  and sparse linear assignment problems,'' \emph{Computing}, vol.~38, no.~4,
  pp. 325--340, 1987.

\end{thebibliography}

\end{document}